\documentclass[prd,nofootinbib,floats,superscriptaddress,eqsecnum,tightenlines,11pt]{revtex4}

\usepackage{hyperref}
\usepackage{graphicx}
\usepackage{diagbox}
\usepackage{amsmath,amssymb,amsfonts,amsthm,latexsym,stmaryrd,mathtools}
\usepackage{dsfont}
\usepackage{pstool}

\def\be#1\ee{\begin{align}#1\end{align}}
\def\bas{\begin{subequations}\begin{eqnarray}}
\def\eas{\end{eqnarray}\end{subequations}}
\def\nn{\nonumber}
\def\sgn{\mathrm{sgn}}

\def\eps{\varepsilon}
\def\lp{\ell_\text{Pl}}
\def\mb{\bar{\mu}}

\def\tr{\mathrm{tr}}

\def\de{\mathrm{d}}
\def\Pexp{\overrightarrow{\exp}}
\def\f{\frac}
\def\lb{\big\lbrace}
\def\rb{\big\rbrace}
\def\SU{\mathrm{SU}}
\def\i{\mathrm{i}}
\def\d{{}^{(d)}\!}

\def\su{\mathfrak{su}}

\def\q{\qquad}

\def\openone{\mathds{1}}

\newtheorem{proposition}{Proposition}

\begin{document}

\title{New Hamiltonians for loop quantum cosmology\\ with arbitrary spin representations}

\author{Jibril Ben Achour }
%\email{jibrilbenachour@gmail.com}
\affiliation{Center for Field Theory and Particle Physics,\\ Fudan University, 200433 Shanghai, China}

\author{Suddhasattwa Brahma}
%\email{suddhasattwa.brahma@gmail.com}
\affiliation{Center for Field Theory and Particle Physics,\\ Fudan University, 200433 Shanghai, China}

\author{Marc Geiller}
%\email{geillermarc@gmail.com}
\affiliation{Perimeter Institute for Theoretical Physics,\\ 31 Caroline Street North, Waterloo, Ontario, Canada N2L 2Y5}

\begin{abstract}	
In loop quantum cosmology, one has to make a choice of SU(2) irreducible representation in which to compute holonomies and regularize the curvature of the connection. The  systematic choice made in the literature is to work in the fundamental representation, and very little is known about the physics associated with higher spin labels. This constitutes an ambiguity whose understanding, we believe, is fundamental for connecting loop quantum cosmology to full theories of quantum gravity like loop quantum gravity, its spin foam formulation, or cosmological group field theory. We take a step in this direction by providing here a new closed formula for the Hamiltonian of flat FLRW models regularized in a representation of arbitrary spin. This expression is furthermore polynomial in the basic variables which correspond to well-defined operators in the quantum theory, takes into account the so-called inverse-volume corrections, and treats in a unified way two different regularization schemes for the curvature. After studying the effective classical dynamics corresponding to single and multiple spin Hamiltonians, we study the behavior of the critical density when the number of representations is increased, and the stability of the difference equations in the quantum theory.
\end{abstract}

\maketitle

%\tableofcontents

\section{Introduction}

\noindent Loop quantum cosmology (LQC hereafter) \cite{B,AS}, a theory attempting to describe the quantum gravitational nature of the early Universe, presents several noteworthy features. First, it leads in a robust way to a resolution of the big-bang singularity in a wide variety of cosmological models \cite{APS,Ash,ACS,BP,PA,APSV,SKL,AWE1,AWE2,WE,BMWE}, and provides, away from the quantum gravity regime, corrections of order $\hbar$ to the standard (e.g. Friedmann) evolution equations \cite{Tav}. Second, it allows for a consistent extension of the framework of cosmological perturbation theory and inflationary physics all the way down to the Plank regime \cite{AAN1,AAN2,AAN3}, and has provided predictions which should be compared to experimental results in the near future\footnote{For an alternative way of treating cosmological perturbations in LQC, see \cite{CovLQC} for an overview and \cite{SigChangeLQC} for consequences.} \cite{AAN3,AG1,AM,Ag,AG2,AAG}. Finally, it has its roots in a full theory of quantum gravity, namely loop quantum gravity (LQG) \cite{lqg1,lqg2,lqg3}.

Although LQC was originally introduced as a minisuperspace quantization, where one imports by hand the key features and technical inputs coming from LQG (such as the existence of an area gap and the representation of the holonomy-flux algebra), lots of recent work has been devoted to identifying a cosmological sector within the full theory, i.e. to properly deriving LQC from (or embedding LQC into) LQG \cite{Bod1,Bod2,Bod3,AC1,AC2,BAC,AC3,AC4,AC5,AC6,AC7,ACM,BEHM,Engle}. In parallel to this, there have also been several advancements on conceptual and technical issues regarding the construction of the continuum limit of LQG and related discrete approaches to quantum gravity. This includes the coarse graining of spin foam models and spin networks \cite{DMBS,BD12b,timeevol,bdreview12,bd14,eckert,BiancaClement,Etera1,Etera2}, and the study of group field theory (GFT) renormalization \cite{C1,COR1,COR2,C2}. In particular, it has recently been shown within GFT that the effective dynamics of condensate states describing the collective behavior of microscopic geometrical degrees of freedom is that of a bouncing cosmological model \cite{GiSi,GOS1,GOS2,OSWE,G1,G2,GO}. GFT being a second quantization of LQG \cite{O1,O2}, this results evidently begs the question of the relationship between LQC and the cosmological sector of GFT.

In any attempt to embed LQC within LQG, or to connect it to GFT, one should understand the role played by possible quantization and regularization ambiguities. In LQC, one such ambiguity appears due to the necessity of choosing a representation of $\SU(2)$, labelled by a half-integer spin $j$, in which to express group elements such as the holonomies, and in which to evaluate the traces defining various quantities like the non-local curvature of the connection expressed in terms these holonomies. The usual choice made in the literature is to work in the fundamental representation of spin $j=1/2$. Similarly, states in LQC are always chosen to be as fine grained as possible, i.e. as being made up of a collection of quanta of geometry each carrying an excitation of spin $j=1/2$ (i.e. the smallest possible area or volume). This is in sharp contrast with the situation in the full theory, where the dynamics (i.e. the Hamiltonian constraint or the spin foam amplitudes) and the states both depend on spin labels. In the study of quantum cosmology via GFT condensates on the other hand, it has been shown that the dynamics selects a low spin regime, and that the LQC states with spin $j=1/2$ are dominant while higher spin contributions are exponentially suppressed \cite{G1}.

The issue of the spin ambiguity in the regularization of the Hamiltonian constraint of LQC has previously been addressed in \cite{V,CL} (and also in \cite{P} in the context of three-dimensional gravity), with differences which we will explain below. In these articles, the formulas obtained by the authors are rather involved. To understand the reason for this, let us recall that the regularization of the Hamiltonian constraint typically requires to write the curvature of the connection in terms of holonomies of this latter around a square plaquette, and then to shrink the area enclosed by this plaquette to a non-zero value given by the area gap of LQG. The existence of this area gap (i.e. the impossibility of taking the limit of vanishing area) is the reason for which the resulting expression for the curvature is non-local and depends on the representation label $j$ in which the holonomies and the trace are computed. In \cite{V,CL}, this is done by computing explicitly the trace in a representation of spin $j$ of group elements expressed in this same representation, which gives rise to the above-mentioned complicated formulas. However, the computation in a representation of spin $j$ of the trace of an $\SU(2)$ group element does not require the ``full information'' about this latter (i.e. the explicit knowledge of all its matrix elements), and can be reduced to an expression involving only the trace in the representation of spin $j=1/2$ (i.e. the fundamental representation) and the class angle. As pointed out in \cite{BAGN}, this simple fact can be used to drastically simplify the expression for the regularized Hamiltonian constraint in LQC with an arbitrary spin representation. This is the main technical feature which we exploit in this note.

In addition to this technical simplification, the improvement which we obtain over \cite{V,CL} also consists in writing in a unified fashion the Hamiltonian constraint operator using the two different regularization schemes which are available in FLRW models for the curvature of the connection, using respectively the holonomies or the connection itself. Following \cite{CK}, we refer to these regularizations as the holonomy and curvature regularizations, and denote them by HR and CR respectively. Also, our construction takes into account the so-called inverse-volume corrections. These results are summarized in formula \eqref{final formula}, which is the expression for the total quantum Hamiltonian constraint operator (with a massless scalar field) taking into account the inverse-volume corrections, the two possible choices HR and CR of regularization for the curvature, and an arbitrary choice of $\SU(2)$ representation.

One of the consequences of considering a representation of spin $j>1/2$ is the obtention of a difference evolution equation for the quantum states which is of higher order in the step size (as compared to the case $j=1/2)$. In \cite{V}, a stability analysis (following the criteria of \cite{BD}) of this higher order difference equation has been performed in the case $j=1$ and HR, and has shown that the quantum theory admits spurious solutions which could spoil its semi-classical limit\footnote{It has also been suggested that these solutions could have zero or infinite norm, and therefore be excluded from the physical Hilbert space. Formula \eqref{final formula} could be taken as a starting point for the construction of this Hilbert space, but we leave this for future work.}. In the present note, we are able to show that the choices $j>1/2$ also fail this stability criterion for the regularization HR, while for the regularization CR only the choices $j>1$ lead to spurious solutions. In addition, the explicit expression for the regularized Hamiltonian constraint and the corresponding classical effective equations of motion show that the regularizations HR and CR do not match for $j>1/2$ (while for $j=1/2$ they can be made to agree by rescaling the parameter $\mb$ entering the regularization).

Interestingly, when considering a representation $j>1/2$ some peculiarities already appear at the classical level. Indeed, following what is usually done in the case $j=1/2$, it is possible to view the regularized Hamiltonian constraint as a heuristic\footnote{For reasons which will becomes clear, we refer to this effective classical Hamiltonian as ``heuristic'' since it is not derived from a rigorous semi-classical analysis of the quantum theory but simply obtained from a regularization of the classical constraint.} effective classical Hamiltonian incorporating holonomy corrections, and then derive effective evolution equations for various physical quantities (such as the Hubble parameter, the energy density, or the volume). By computing the evolution of the matter density as a function of the Hubble parameter, we find that for $(j>1/2,\text{HR})$ and $(j>1,\text{CR})$ the density becomes negative during its evolution. This means that the corresponding heuristic Hamiltonian constraints do not describe proper semi-classical FLRW dynamics. At first, one could think that this is the manifestation at the classical level of the fact that there exist spurious unstable solutions in the quantum theory. However, the analysis of extended Hamiltonian constraints which take into account the presence of multiple spins shows that this relation is not straightforward. As we will show, it is possible to consider arbitrary linear combinations of Hamiltonian constraints which are regularized with different $\SU(2)$ representations. If one sums an arbitrary number of successive Hamiltonians with spins starting from $j=1/2$, for a very generic class of weights the resulting evolution equations lead to a bouncing FLRW behavior. Interestingly, if the weights assigned to higher spins are not decreasing fast enough, the critical density at which the bounce occurs tends to zero. On the other hand, for weights which are decreasing sufficiently fast (typically exponentially), as the number of allowed spins is increased the critical density converges towards a value which is always lower than the usual $j=1/2$ value $\rho_\text{c}\approx0.41\rho_\text{Pl}$. However, in spite of the density remaining positive as a function of the Hubble rate when one considers the contribution of several spins, the corresponding quantum Hamiltonian constraint still admits spurious solutions. This shows that there is indeed no a priori relation between the peculiar behavior of the heuristic classical equations and the existence of spurious solutions in the quantum theory. It is interesting to note however that the relationship between the lowering of the critical density and the inclusion of exponentially suppressed spins $j>1/2$ is reminiscent of what has been observed in cosmological applications of GFT \cite{G1}.

This note is organized as follows. In section \ref{sec:2}, we briefly review the construction of flat FLRW LQC and its effective dynamics in the case $j=1/2$. This is the occasion of introducing the various notations used throughout the rest of this work, and will serve as a comparison with the case of higher representation labels. In section \ref{sec:3}, we construct the two different regularizations HR and CR of the curvature operator for higher spins. In section \ref{sec:4}, we study the heuristic classical equations of motion which are obtained by considering the regularized Hamiltonian as a classical constraint. In section \ref{sec:5} we incorporate the inverse-volume corrections, and finally in section \ref{sec:6} we briefly comment on the quantum theory.

\section{Review of flat homogeneous isotropic LQC}
\label{sec:2}

\noindent In order to make this article self-contained and to set some standard notations and conventions, we briefly review in this section some well-known features about flat, homogeneous, and isotropic LQC. The reader already familiar with LQC can safely skip this part and jump to section \ref{sec:3}.

In the $\SU(2)$ Ashtekar--Barbero formulation of first order gravity, the action takes the $3+1$ Hamiltonian form
\be
S[E^a_i,A^i_a,N,N^a,A^i_0]=\int_\mathbb{R}\de t\int_\Sigma\de^3x\,\left(\f{1}{\kappa\gamma}E^a_i\partial_0{A}^i_a-NC_\text{g}-N^aC_a-A^i_0G_i\right).
\ee
Here, and in the rest is this work, we use $a,b,c\dots$ to denote spatial indices and $i,j,k,\dots$ to denote $\su(2)$ Lie algebra indices. The canonical phase space variables are the connection and the densitized triad, whose explicit expression is given respectively by
\be
A^i_a\coloneqq\Gamma^i_a+\gamma K^i_a,\q\q E^a_i\coloneqq\det(e^i_a)e^a_i=\f{1}{2}\eps^{abc}\eps_{ijk}e^j_be^k_c,
\ee
and these variables obey the Poisson bracket relation
\be\label{AE Poisson}
\lb A^i_a(x),E^b_j(y)\rb=\kappa\gamma\delta^i_j\delta^b_a\delta^3(x,y).
\ee
Here $\Gamma^i_a$ is the Levi--Civita spin connection, $K^i_a$ is the extrinsic curvature, $\kappa\coloneqq8\pi G$, and $\gamma$ is the Barbero--Immirzi parameter. As usual, the total Hamiltonian is the sum of the scalar, vector, and Gauss constraints, which are enforced respectively by the lapse $N$, the shift vector $N^a$, and the connection component $A^i_0$. Upon imposition of the cosmological symmetries (i.e. homogeneity and isotropy), only the scalar constraint will be non-vanishing. For our purposes, it is therefore sufficient to only focus on the explicit expression of the (spatially-integrated) scalar constraint, which is given by\footnote{The subscript refers here to the gravitational part of the scalar constraint, and later on we will consider the addition of a matter part.}
\be\label{full scalar constraint}
\mathcal{C}_\text{g}^N\coloneqq\int_\Sigma\de^3x\,NC_\text{g}=\f{1}{2\kappa}\int_\Sigma\de^3x\,Nq^{-1/2}E^a_iE^b_j\left({\eps^{ij}}_kF^k_{ab}-2\left(1+\gamma^2\right)K^i_{[a}K^j_{b]}\right),
\ee
where $q^{1/2}\coloneqq|\det(q_{ab})|^{1/2}=|\det(E^a_i)|^{1/2}=|\det(e^i_a)|$ is the square root of the determinant of the spatial metric. Since we are going to work later on with different choices for the lapse $N$, we have indicated in $\mathcal{C}_\text{g}^N$ an explicit dependency on $N$ in order to keep track of which choice has been made.

\subsection{Symmetry reduction}

\noindent We are now going to restrict ourselves to a spatially flat, homogeneous, and isotropic spacetime, i.e. to a flat FLRW cosmological model. In this case, the line element in comoving coordinates $x^a=(x,y,z)$ takes the form
\be
\de s^2=-\de t^2+q_{ab}\de x^a\de x^b=-\de t^2+a^2(\de x^2+\de y^2+\de z^2)=-\de t^2+a^2\mathring{q}_{ab}\de x^a\de x^b,
\ee
where $a=a(t)$ is the scale factor of the Universe and $t$ is the proper time along the worldlines of observers moving in a direction orthogonal to that of the spatial slices. In this expression, $q_{ab}$ is the physical metric of the homogeneous spatial slices $\Sigma$, and for later purposes we have used the comoving coordinates to define a fiducial flat background metric
\be
\mathring{q}_{ab}\de x^a\de x^b=\de x^2+\de y^2+\de z^2,\q\q\mathring{q}_{ab}=\mathring{e}^i_a\mathring{e}^j_b\delta_{ij},
\ee
where $\mathring{e}^i_a$ is a fixed fiducial cotriad. The physical spatial three-dimensional metric is related to this fiducial metric through the scale factor by $q_{ab}=a^2\mathring{q}_{ab}$. As we shall see later on, it turns out to be useful in LQC to work with the harmonic time coordinate $\tau$, which is such that $\Box\tau=0$, and in terms of which the line element takes the form
\be
\de s^2=-a^6\de\tau^2+q_{ab}\de x^a\de x^b.
\ee
This choice corresponds in fact to setting the lapse, defined via $N\de\tau=\de t$, to $N=a^3$.

In flat FLRW cosmology, the topology of the spatial manifold $\Sigma$ can be chosen to be either that of a three-torus, or that of $\mathbb{R}^3$. Let us focus here on the case $\Sigma=\mathbb{R}^3$. Since the spatial manifold is then non-compact, all spatial integrals are divergent, and the Lagrangian, the Hamiltonian, and the symplectic structure are a priori ill-defined. This can easily be dealt with by restricting all spatial integrations to a fixed fiducial cell $\mathcal{V}$, which we can choose to be cubical with respect to every physical metric $q_{ab}$ on $\mathbb{R}^3$. More precisely, using the cell $\mathcal{V}$ as a regulator for the spatial integrals means that we have to proceed to the substitution
\be
\int_\Sigma\de^3x\rightarrow\int_\mathcal{V}\de^3x.
\ee
We will denote the fiducial volume of this cell, defined as the volume measured with respect to the fiducial metric $\mathring{q}_{ab}$, by
\be
\mathring{V}\coloneqq\int_\mathcal{V}\de^3x\,\mathring{q}^{1/2},
\ee
where $\mathring{q}^{1/2}\coloneqq|\det(\mathring{q}_{ab})|^{1/2}$. Note that there is always a freedom in rescaling the fiducial metric and the scale factor as $(\mathring{q}_{ab},a)\rightarrow(\alpha^2\mathring{q}_{ab},\alpha^{-1}a)$, as well a freedom in rescaling the cell as $\mathcal{V}\rightarrow\beta^3\mathcal{V}$. Since this freedom cannot be gauged away (unlike in the case of positive or negative spatial curvature, where one can choose a unit metric on $\Sigma$), one has to remember that the scale factor $a$ has no direct physical meaning. Only ratios of scale factors do carry a physical interpretation.

Now that we have appropriately regularized the infinite volume of the spatial manifold, we are ready to see how the cosmological symmetries affect the phase space variables of the theory. Isotropic connections and densitized triads are characterized by
\be
A^i_a=\tilde{c}\mathring{e}^i_a\rightarrow c\mathring{V}^{-1/3}\mathring{e}^i_a,\q\q E^a_i=\tilde{p}\mathring{q}^{1/2}\mathring{e}^a_i\rightarrow p\mathring{q}^{1/2}\mathring{V}^{-2/3}\mathring{e}^a_i,
\ee
where $\mathring{e}^a_i$ the fiducial triad adapted to the edges of $\mathcal{V}$. Here we have used the freedom in redefining the fiducial metric in order to rescale the symmetry-reduced variables as $(\tilde{c},\tilde{p})\rightarrow(c\mathring{V}^{-1/3},p\mathring{V}^{-2/3})$. This ensures that the resulting symmetry-reduced Poisson bracket will be independent of $\mathring{q}_{ab}$ and $\mathcal{V}$. With these new cosmological variables, the Poisson bracket \eqref{AE Poisson} reduces to
\be
\lb c,p\rb=\f{\kappa\gamma}{3}.
\ee
As a remark, let us point out that although the symplectic structure and the scalar constraint will turn out to be independent of the choice of fiducial metric $\mathring{q}_{ab}$, they will still both carry a hidden dependency on the choice of fiducial cell.

The dynamics of LQC is usually formulated in terms of the above symmetry-reduced connection and triad variables. However, it will also be helpful later on to use the variables
\be\label{b and v variables}
b\coloneqq c|p|^{-1/2},\q\q\nu\coloneqq\f{\sgn(p)}{2\pi\gamma\lp^2}|p|^{3/2},
\ee
which satisfy the Poisson bracket
\be
\lb b,\nu\rb=\f{2}{\hbar}.
\ee
It is also informative to relate these variables to the scale factor. The relation between the scale factor and the above variables is as follows\footnote{Note that since the relations with $\dot{a}$ involve the dynamics, i.e. the time evolution of $a$, these relations are only true in the classical Hamiltonian framework, and simply express the fact that we are free to work with the phase space variables $(a,\dot{a})$, $(c,p)$ or $(b,\nu)$. As we will see later on, these relations are modified in LQC (whether it is in the full quantum theory or at the level of the effective classical equations) since the dynamics is then different.}:
\be
c=\gamma\dot{a}\mathring{V}^{1/3},\q\q|p|=a^2\mathring{V}^{2/3},\q\q b=\gamma\f{\dot{a}}{a},
\ee
where the dot indicates the derivative with respect to proper (or cosmic) time $t$. We therefore see that $p$ measures the physical area of each rectangular face of the cell $\mathcal{V}$, while $b$ is proportional to the Hubble rate. Note that we also have the relations
\be
q^{1/2}=a^3\mathring{q}^{1/2}=|p|^{3/2}\mathring{V}^{-1}\mathring{q}^{1/2},
\ee
and that we can therefore express the physical volume of $\mathcal{V}$ as
\be
V=a^3\mathring{V}=\int_\mathcal{V}\de^3x\,q^{1/2}=|p|^{3/2}.
\ee
Finally, note that we can express the physical triad and cotriad as
\be
e^a_i=\sgn(p)|p|^{-1/2}\mathring{V}^{1/3}\mathring{e}^a_i,\q\q e^i_a=\sgn(p)|p|^{1/2}\mathring{V}^{-1/3}\mathring{e}^i_a.
\ee
The variable $p$ can be either positive or negative, which corresponds to the two possible choices of orientation for the triad. Since the metric is left unchanged under the transformation $e^a_i\mapsto-e^a_i$ and we will not consider the coupling to fermions, the physics will also be left unchanged by this transformation. Thus, this orientation reversal of the triad represents a gauge transformation which has to be dealt with in the quantum theory.

We can now focus on the scalar constraint of the theory. Because the model that we describe is spatially flat, the spin connection $\Gamma^i_a$ vanishes and we have $A^i_a=\gamma K^i_a$. Furthermore, because of spatial homogeneity, the derivatives of $K^i_a$ do also vanish and we have that
\be
2K^i_{[a}K^j_{b]}=\frac{1}{\gamma^2}{\eps^{ij}}_kF^k_{ab}.
\ee
With this, we find that the gravitational part \eqref{full scalar constraint} of the scalar constraint can be expressed as
\be\label{reduced gravitational hamiltonian}
\mathcal{C}_\text{g}^N=-\f{1}{2\kappa\gamma^2}N|p|^{1/2}\mathring{V}^{2/3}\mathring{e}^a_i\mathring{e}^b_j{\eps^{ij}}_kF^k_{ab}=-\f{3}{\kappa\gamma^2}N|p|^{1/2}c^2,
\ee
where we have written the intermediate expression for later use, as well as an explicit dependency on the lapse. The simplest type of matter coupling which we could consider is the addition of a free massless scalar field. The contribution of such a matter sector to the scalar constraint is given by
\be\label{matter hamiltonian}
\mathcal{C}_\text{m}^N=\f{1}{2}\int_\mathcal{V}\de^3x\,Nq^{1/2}\dot{\phi}^2=\f{1}{2}N|p|^{3/2}\dot{\phi}^2=\f{1}{2}N|p|^{-3/2}p_\phi^2,
\ee
where the momentum conjugate to the scalar field is $p_\phi=|p|^{3/2}\dot{\phi}$, and we have $\lb\phi,p_\phi\rb=1$. In summary, the total scalar constraint (which is in fact the total Hamiltonian of the symmetry-reduced cosmological model) takes the form
\be
\mathcal{C}^N\coloneqq\mathcal{C}_\text{g}^N+\mathcal{C}_\text{m}^N=N\left(-\f{3}{\kappa\gamma^2}|p|^{1/2}c^2+\f{1}{2}|p|^{-3/2}p_\phi^2\right),
\ee
In order to go back to more conventional variables, we can reintroduce the scale factor along with its momentum $p_a=-a\dot{a}$, and use $p_\phi=V\dot{\phi}$ with $V=a^3\mathring{V}=|p|^{3/2}$. This leads to
\be
\mathcal{C}^N=N\left(-\f{3}{\kappa}\f{Vp_a^2}{a^4}+\f{1}{2}\f{p_\phi^2}{V}\right).
\ee
When the lapse is chosen to be $N=1$, the evolution in proper time is given by $\partial_t(\cdot)=\lb\cdot,\mathcal{C}^1\rb$, and one can check that the equation $\dot{p}_a=\lb p_a,\mathcal{C}^1\rb$, together with the vanishing of the Hamiltonian constraint, leads to the Friedmann equation
\be
-\f{1}{2}\f{\ddot{a}}{a}=\left(\f{\dot{a}}{a}\right)^2=\f{\kappa}{3}\rho,
\ee
where
\be\label{classical density}
\rho\coloneqq\f{1}{2}|p|^{-3}p_\phi^2=\f{1}{2}\dot{\phi}^2
\ee
is the energy density of the scalar field.

So far, we have defined the Hamiltonian dynamics of the symmetry-reduced variables. However, in LQG and LQC the connection cannot be associated to a well-defined operator in the quantum theory, and one has to rewrite the constraint in terms of the holonomies of the connection, which are the only well-defined operators. In particular, one has to proceed to a holonomy regularization of the curvature operator $F^i_{ab}$. This is precisely the step which depends on the choice of spin representation.

\subsection{Regularization of the curvature in the spin 1/2 representation}

\noindent Recasting the Hamiltonian constraint in a form which is suited for quantization via the techniques of LQG requires essentially two steps. The first one is to regularize the curvature of the connection in terms of the holonomies of this latter, and the second one is to appropriately deal with the presence of an inverse power of the determinant of the triad. The first step leads to the so-called holonomy corrections, while the second one, when carried out using Thiemann's trick, leads to the so-called inverse-volume corrections. It is however known that the quantum gravity effects which are responsible for the resolution of the singularity in LQC can ``mostly'' be traced back to the holonomy corrections, while the inverse-volume corrections only play an auxiliary role. Therefore, as a simplification which will however enable us to encode the relevant physics, let us get rid of the inverse determinant of the triad in the Hamiltonian constraint by simply going to the harmonic gauge, i.e. by choosing\footnote{Since $p$ represents the physical area of the faces of the cell $\mathcal{V}$, it is independent of the coordinates and transforms as a scalar. On the other hand, $q^{1/2}=|p|^{3/2}\mathring{V}^{-1}\mathring{q}^{1/2}$ is a scalar density of weight one because of the coordinate-dependent factors. Since $N$ is itself a scalar, this explains why in LQC it is possible and meaningful to choose $N=|p|^{3/2}$, while in the full theory it is not possible to choose $N=q^{1/2}$.} the lapse to be $N=|p|^{3/2}$.

In the harmonic gauge, which we denote by using the superscript ``h'' in the scalar constraint, the gravitational part \eqref{reduced gravitational hamiltonian} of the Hamiltonian constraint becomes
\be\label{harmonic gauge gravitational hamiltonian}
\mathcal{C}_\text{g}^\text{h}=-\frac{1}{2\kappa\gamma^2}|p|^2\mathring{V}^{2/3}\mathring{e}^a_i\mathring{e}^b_j{\eps^{ij}}_kF^k_{ab}.
\ee
The first step towards regularizing the curvature of the connection is to compute the holonomies of this latter. For this, let us consider an edge aligned with the direction of $\mathring{e}^a_i\partial_a$, and of fiducial length $\mb\mathring{V}^{1/3}$. In the fundamental representation of dimension $d=2$ (i.e. of spin $j=1/2$), the holonomy of the connection along this edge is given by
\be
h_{i^\epsilon}^{(\mb)}=\Pexp\int_0^{\mb\mathring{V}^{1/3}}\epsilon A^i_a\tau_i\de x^a=\exp(\epsilon\mb c\tau_i)=\cos\left(\f{\mb c}{2}\right)\openone+2\epsilon\sin\left(\f{\mb c}{2}\right)\tau_i,
\ee
where $\epsilon=\pm1$ enables us to choose between the holonomy along the edge $i$ and its inverse, and $\tau_i$ is a generator of $\su(2)$ (see appendix \ref{appendixA}). In what follows, we will denote by $h_{i^{-1}}^{(\mb)}$ the inverse of the group element $h_i^{(\mb)}$. Now, we can introduce the holonomy of the connection around a square plaquette $\Box_{ij}$ perpendicular to $\mathring{e}^a_k\partial_a$ and of fiducial area $\text{Ar}_\Box=\mb^2\mathring{V}^{2/3}$. This is given by
\be
h_{\Box_{ij}}^{(\mb)}=h_i^{(\mb)}h_j^{(\mb)}h_{i^{-1}}^{(\mb)}h_{j^{-1}}^{(\mb)}.
\ee
In terms of this plaquette holonomy, we can then write the following standard expression for the curvature:
\be
F^k_{ab}
&=\f{1}{\tau(2)}\lim_{\text{Ar}_\Box\rightarrow0}\f{1}{\text{Ar}_\Box}\tr\left(\left[h_{\Box_{ij}}^{(\mb)}-\openone\right]\tau^k\right)\mathring{e}^i_a\mathring{e}^j_b\nn\\
&=\f{1}{\tau(2)}\lim_{\mb\rightarrow0}\f{1}{\mb^2\mathring{V}^{2/3}}\tr\left(h_{\Box_{ij}}^{(\mb)}\tau^k\right)\mathring{e}^i_a\mathring{e}^j_b,
\ee
where the trace is taken in the representation of dimension $d=2$ with the normalization factor \eqref{trace normalization}. In the second line, we have simply used the fact that, by definition, the generators of the Lie algebra $\su(2)$ have vanishing trace, and used the expression for the area of the plaquette in terms of $\mb$ and $\mathring{V}$.

At this point, the crucial input from the full theory is to realize that this limit is in fact not well-defined. Physically speaking, this is due to the presence of a minimal non-zero eigenvalue for the area operator, which implies that one can only shrink $\text{Ar}_\Box$ down to the so-called area gap of LQG (which, as explained in the next subsection, is proportional to $\lp^2$). In mathematical terms, this non-existence of the limit is due to the fact that the basic phase space functions, namely the holonomies of the connection $c$, are not weakly-continuous with respect to the parameter $\mb$ measuring the edge length. At this point, we are therefore going to approximate the curvature by
\be
F^k_{ab}\approx-\f{2}{\mb^2\mathring{V}^{2/3}}\tr\left(h_{\Box_{ij}}^{(\mb)}\tau^k\right)\mathring{e}^i_a\mathring{e}^j_b,
\ee
and the explicit relation between the parameter $\mb$ and the area gap will be given later on. Now, using the explicit expression
\be\label{plaquette holo}
h_{\Box_{ij}}^{(\mb)}=\left(\cos(\mb c)+\f{1}{2}\sin^2(\mb c)\right)\openone+\sin^2(\mb c){\eps_{ij}}^k\tau_k+2\sin^2\left(\f{\mb c}{2}\right)\sin(\mb c)(\tau_i-\tau_j),
\ee
one finds
\be
\tr\left(h_{\Box_{ij}}^{(\mb)}\tau^k\right)=-\f{1}{2}\sin^2(\mb c){\eps_{ij}}^k,
\ee
and the scalar constraint \eqref{harmonic gauge gravitational hamiltonian} in the harmonic gauge finally becomes\footnote{We now write strict equalities instead of approximate ones.}
\be\label{regularized constraint in d=2}
{}^{(2)}\mathcal{C}_\text{g}^\text{h}=-\f{3}{\kappa\gamma^2\mb^2}|p|^2\sin^2(\mb c).
\ee
Here we have added a superscript $(2)$ in order to remember that this expression has been obtained by regularizing the curvature operator in the representation of dimension $d=2$. This expression is the (simplest) regularized version of the Hamiltonian constraint which is used in LQC. It is written in terms of quantities which are all well-defined on the kinematical Hilbert space of the theory, and, modulo a choice of factor ordering, one can therefore go on to the study of the physical states and the quantum dynamics.

Now, one last subtlety is to determine an explicit expression for the parameter $\mb$. Again, this is done by importing ingredients from the full theory. Consider one of the square faces defining the boundary of the cubical cell $\mathcal{V}$. Assuming that this face is pierced transversally by $N$ spin network links carrying the minimal possible spin label, i.e. $j=1/2$ (which corresponds to the best possible coarse-grained homogeneity), one can subdivide the face into $N$ identical cells which are each pierced by a single spin network link. Each of these elementary cells carries an area determined by the area gap of LQG, namely
\be\label{lambda}
\lambda^2\coloneqq4\sqrt{3}\pi\gamma\lp^2,
\ee
which is the smallest non-vanishing eigenvalue of the area spectrum
\be
A(j)=\kappa\gamma\hbar\sqrt{j(j+1)}=8\pi\gamma\lp^2\sqrt{j(j+1)}=4\pi\gamma\lp^2\sqrt{d^2-1}.
\ee
Therefore, the face has a total physical area given by
\be
|p|=N\lambda^2.
\ee
On the other hand, recall that the parameter measuring the length of the edges of each elementary cell is $\mb$. In terms of the fiducial metric, each cell has an area $\mb^2\mathring{V}^{2/3}$. Since the total fiducial area of the face is $\mathring{V}^{2/3}$, we get the equality
\be
N\mb^2\mathring{V}^{2/3}=\mathring{V}^{2/3}.
\ee
Putting this together with the previous equation we finally obtain the relation
\be
\mb^2|p|=\lambda^2,
\ee
or in other words $\mb=\lambda|p|^{-1/2}$. This shows that $\mb$ contains in fact a non-trivial dependency on the dynamical variables.

\subsection{Effective dynamics}
\label{effective dynamics}

\noindent The effective framework of LQC is a classical Hamiltonian dynamics which encodes the leading order contributions from the quantum gravity effects. Rigorously speaking, this effective dynamics has to be extracted from the quantum dynamics by taking an appropriate semi-classical limit and using coherent states. One simpler and heuristic way to proceed is simply to consider the expression \eqref{regularized constraint in d=2} coming from the regularization of the curvature instead of the original scalar constraint\footnote{Note that \eqref{regularized constraint in d=2} can be obtained from \eqref{reduced gravitational hamiltonian} by the so-called ``polymerisation'' of the connection variable, which consists in replacing $c\rightarrow\sin(\mb c)/\mb$. This fact is, however, a particularity of flat FLRW models, and does not hold in general (we thank Edward Wilson-Ewing for bringing this to our attention).} \eqref{reduced gravitational hamiltonian}.

In order to derive the effective Friedmann equation, let us work with proper time\footnote{Since we are here only considering the heuristic modification of the constraint including holonomies, we allow ourselves to freely change between the harmonic gauge and the proper time gauge. In the quantum theory, this is of course not allowed, since working with proper time requires the inclusion of the inverse-volume corrections.}, i.e. with a lapse $N=1$, and consider the heuristic modification of the Hamiltonian constraint given by
\be
{}^{(2)}\mathcal{C}^1={}^{(2)}\mathcal{C}^1_\text{g}+\mathcal{C}^1_\text{m}=-\f{3}{\kappa\gamma^2\mb^2}|p|^{1/2}\sin^2(\mb c)+|p|^{3/2}\rho.
\ee
One can then compute the equation of motion
\be
\dot{p}=\lb p,{}^{(2)}\mathcal{C}^1\rb=\f{2}{\gamma\mb}|p|^{1/2}\sin(\mb c)\cos(\mb c),
\ee
which in turn leads to
\be
\left(\f{\dot{a}}{a}\right)^2=\f{1}{4}\left(\f{\dot{p}}{p}\right)^2=\f{1}{\gamma^2\mb^2}|p|^{-1}\sin^2(\mb c)\cos^2(\mb c)=\f{1}{\gamma^2\mb^2}|p|^{-1}\sin^2(\mb c)\big(1-\sin^2(\mb c)\big).
\ee
Finally, using the vanishing of the effective Hamiltonian constraint, we get the modified Friedmann equation
\be
H^2=\left(\f{\dot{a}}{a}\right)^2=\f{\kappa}{3}\rho\left(1-\f{\rho}{\rho_\text{c}}\right),
\ee
where the critical density is given by
\be\label{critical density}
\rho_\text{c}\coloneqq\f{3}{\kappa\gamma^2\mb^2}|p|^{-1}=\f{3}{\kappa\gamma^2\lambda^2}=\f{\sqrt{3}}{4\pi\kappa\gamma^3\lp^2}=\f{\sqrt{3}}{32\pi^2\gamma^3}\rho_\text{Pl},
\ee
and the Planck density is $\rho_\text{Pl}=(\hbar G^2)^{-1}$ in units in which the speed of light is 1. Using the value $\gamma\approx0.2375$ for the Barbero--Immirzi parameter\footnote{This value, which is often used in LQC and for example in recent phenomenological investigations \cite{AG2}, comes from the requirement that the black hole state counting in LQG leads to the celebrated Bekenstein--Hawking formula $S=A/(4\lp^2)$ \cite{BH1,BH2,BH3,BH4,BH5,EPN,ENPP,PP,LT1,LT2}. Although there is some technical freedom in performing this calculation, which can lead to different values of $\gamma$, one can however always expect $\gamma$ to be of order one.}, one finds a value for the critical density which is $\rho_\text{c}\approx0.41\rho_\text{Pl}$.

\section{Regularizations of the curvature operator}
\label{sec:3}

\noindent In this section, we present the two regularizations of the curvature operator for arbitrary $\SU(2)$ representations. Following \cite{CK}, we denote these two regularization schemes by HR (for holonomy regularization) and CR (for connection regularization). As we will see, the resulting regularized curvature operators in these two cases can be written in the exact same form, and the information about the regularization scheme is simply contained in the class angle of a certain $\SU(2)$ group element.

\subsection{Holonomy regularization}

\noindent The holonomy regularization is the one which we have used in the previous section in the case $d=2$. We explain here how the simple formula of appendix \ref{appendixA} can be used to generalize this result to the case $d>2$. To compute the curvature, one considers again a square plaquette $\Box_{ij}$ perpendicular to the direction of $\mathring{e}^a_k\partial_a$, with edges of coordinate length $\mb\mathring{V}^{1/3}$, and of area $\text{Ar}_\Box=\mb^2\mathring{V}^{2/3}$. Then we can write that\footnote{From now on we will simply remove the limit from various expressions, and still use an equal sign.}
\be\label{curvature holo}
\d F^k_{ab}
&=\f{1}{\tau(d)}\lim_{\text{Ar}_\Box\rightarrow0}\f{1}{\text{Ar}_\Box}\tr_{(d)}\left(\left[\d h^{(\mb)}_{\Box_{ij}}-\openone\right]\d\tau^k\right)\mathring{e}^i_a\mathring{e}^j_b\nn\\
&=\f{1}{\tau(d)}\f{1}{\mb^2\mathring{V}^{2/3}}\tr_{(d)}\left(\d h^{(\mb)}_{\Box_{ij}}\d\tau^k\right)\mathring{e}^i_a\mathring{e}^j_b,
\ee
with
\be
\d h^{(\mb)}_{\Box_{ij}}=\d h_i^{(\mb)}\d h_j^{(\mb)}\d h_{i^{-1}}^{(\mb)}\d h_{j^{-1}}^{(\mb)}.
\ee
Since this plaquette holonomy is obviously a group element, one can rewrite the trace in \eqref{curvature holo} using the strategy explained in appendix \ref{appendixA} and summarized in formula \eqref{trace formula}.

For this, we first need to compute the trace in the fundamental representation. This has already been done above using \eqref{plaquette holo}, and one simply has that
\be
\tr\left(h^{(\mb)}_{\Box_{ij}}\tau^k\right)=-\f{1}{2}\sin^2(\mb c){\eps_{ij}}^k.
\ee
Formula \eqref{trace formula} then leads to
\be
\tr_{(d)}\left(\d h^{(\mb)}_{\Box_{ij}}\d\tau^k\right)=\f{\sin^2(\mb c)}{4\sin\theta}\f{\partial}{\partial\theta}\left(\f{\sin(d\theta)}{\sin\theta}\right){\eps_{ij}}^k,
\ee
where the class angle, most easily computed in the fundamental representation, is given by
\be\label{class angle 1}
\theta=\arccos\left(\f{1}{2}\tr\left(h^{(\mb)}_{\Box_{ij}}\right)\right)=\arccos\left(\cos(\mb c)+\f{1}{2}\sin^2(\mb c)\right).
\ee
Putting this together with \eqref{curvature holo}, we finally get that the curvature operator can be approximated by
\be\label{curvature for dimension d}
\d F^k_{ab}=-\f{3}{d(d^2-1)}\f{1}{\mb^2\mathring{V}^{2/3}}\f{\sin^2(\mb c)}{\sin\theta}\f{\partial}{\partial\theta}\left(\f{\sin(d\theta)}{\sin\theta}\right){\eps_{ij}}^k\mathring{e}^i_a\mathring{e}^j_b.
\ee
From this, we can now obtain for example
\begin{subequations}
\be
{}^{(2)\!}F^k_{ab}&=\f{1}{\mb^2\mathring{V}^{2/3}}\sin^2(\mb c){\eps_{ij}}^k\mathring{e}^i_a\mathring{e}^j_b,\\
{}^{(3)\!}F^k_{ab}&=\f{1}{\mb^2\mathring{V}^{2/3}}\sin^2(\mb c)\f{1}{2}\big(2\cos(\mb c)+\sin^2(\mb c)\big){\eps_{ij}}^k\mathring{e}^i_a\mathring{e}^j_b,\\
{}^{(4)\!}F^k_{ab}&=\f{1}{\mb^2\mathring{V}^{2/3}}\sin^2(\mb c)\f{1}{10}\big(10-12\sin^2(\mb c)+12\sin^2(\mb c)\cos(\mb c)+3\sin^4(\mb c)\big){\eps_{ij}}^k\mathring{e}^i_a\mathring{e}^j_b.
\ee
\end{subequations}

\subsection{Connection regularization}

\noindent Because of spatial homogeneity, there is another useful way of writing the curvature operator. The reason is that $F^k_{ab}={\eps^k}_{ij}A^i_aA^j_b$ because of the vanishing of the spatial derivatives, and the connection can be expressed via the formula
\be\label{connection as function of holonomy}
A^j_a\tau_j\mathring{e}^a_i=\lim_{\mb\rightarrow0}\f{1}{4\mb\mathring{V}^{1/3}}\left(\d h_i^{(2\mb)}-\d h_{i^{-1}}^{(2\mb)}\right).
\ee
Here the factor of 2 multiplying $\mb$ has been introduced for simplifications that will occur later on. In particular, we will see that this factor of 2 is necessary for the regularization schemes HR and CR to agree in the case $d=2$. For $d>2$ however, the two schemes will always lead to different results.

Using \eqref{connection as function of holonomy}, the curvature can be approximated by
\be\label{curvature holo 2}
\d F^k_{ab}&=\f{1}{\tau(d)}\lim_{\mb\rightarrow0}\f{1}{16\mb^2\mathring{V}^{2/3}}\tr_{(d)}\left(\Big[\d h_i^{(2\mb)}-\d h_{i^{-1}}^{(2\mb)},\d h_j^{(2\mb)}-\d h_{j^{-1}}^{(2\mb)}\Big]\d\tau^k\right)\mathring{e}^i_a\mathring{e}^j_b\nn\\
&=\f{1}{\tau(d)}\lim_{\text{Ar}_\Box\rightarrow0}\f{1}{16\text{Ar}_\Box}\tr_{(d)}\left(\Big[\d h_i^{(2\mb)}-\d h_{i^{-1}}^{(2\mb)},\d h_j^{(2\mb)}-\d h_{j^{-1}}^{(2\mb)}\Big]\d\tau^k\right)\mathring{e}^i_a\mathring{e}^j_b\nn\\
&=\f{1}{\tau(d)}\f{1}{16\mb^2\mathring{V}^{2/3}}\tr_{(d)}\left(\Big[\d h_i^{(2\mb)}-\d h_{i^{-1}}^{(2\mb)},\d h_j^{(2\mb)}-\d h_{j^{-1}}^{(2\mb)}\Big]\d\tau^k\right)\mathring{e}^i_a\mathring{e}^j_b.
\ee
Notice that although we have not introduced a square plaquette $\Box$ for the regularization of the curvature (as we did previously in the scheme HR), in the second equality we have explicitly rewritten $\mb^2\mathring{V}^{2/3}$ as $\text{Ar}_\Box$. The reason for this is that, following the full theory, we need to argue that the area gap prevents us from taking the limit $\text{Ar}_\Box\rightarrow0$. However, since in LQG there is no length operator admitting a length gap, this reasoning cannot be made on the limit $\mb\rightarrow0$ and has to be done on the limit $\text{Ar}_\Box\rightarrow0$ instead.

Now, another way to write the expression \eqref{curvature holo 2} is to expand to commutator and to write the total trace as a sum of eight traces of the type
\be
\tr_{(d)}\left(\d h_{i^\epsilon}^{(2\mb)}\d h_{j^\sigma}^{(2\mb)}\d\tau^k\right),
\ee
where $\epsilon$ and $\sigma$ are equal to $\pm1$ depending on the term under consideration. Since this is the trace of the product of a group element with a Lie algebra element, we can again apply formula \eqref{trace formula}. More precisely, one can show that
\be
\tr\left(h_{i^\epsilon}^{(2\mb)}h_{j^\sigma}^{(2\mb)}\tau^k\right)=-\epsilon\sigma\sin^2(\mb c){\eps_{ij}}^k,
\ee
and use formula \eqref{trace formula} to write
\be
\tr_{(d)}\left(\Big[\d h_i^{(2\mb)}-\d h_{i^{-1}}^{(2\mb)},\d h_j^{(2\mb)}-\d h_{j^{-1}}^{(2\mb)}\Big]\d\tau^k\right)=4\f{\sin^2(\mb c)}{\sin\theta}\f{\partial}{\partial\theta}\left(\f{\sin(d\theta)}{\sin\theta}\right){\eps_{ij}}^k.
\ee
The class angle appearing in this formula is given by
\be\label{class angle 2}
\theta=\arccos\left(\f{1}{2}\tr\left(h_{i^\epsilon}^{(2\mb)}h_{j^\sigma}^{(2\mb)}\right)\right)=\arccos\big(\cos^2(\mb c)\big).
\ee
Finally, for the curvature operator we find once again the same formula, i.e.
\be
\d F^k_{ab}=-\f{3}{d(d^2-1)}\f{1}{\mb^2\mathring{V}^{2/3}}\f{\sin^2(\mb c)}{\sin\theta}\f{\partial}{\partial\theta}\left(\f{\sin(d\theta)}{\sin\theta}\right){\eps_{ij}}^k\mathring{e}^i_a\mathring{e}^j_b,
\ee
except that $\theta$ is now given by \eqref{class angle 2}. From this, we can now obtain
\begin{subequations}
\be
{}^{(2)\!}F^k_{ab}&=\f{1}{\mb^2\mathring{V}^{2/3}}\sin^2(\mb c){\eps_{ij}}^k\mathring{e}^i_a\mathring{e}^j_b,\\
{}^{(3)\!}F^k_{ab}&=\f{1}{\mb^2\mathring{V}^{2/3}}\sin^2(\mb c)\cos^2(\mb c){\eps_{ij}}^k\mathring{e}^i_a\mathring{e}^j_b,\\
{}^{(4)\!}F^k_{ab}&=\f{1}{\mb^2\mathring{V}^{2/3}}\sin^2(\mb c)\f{1}{5}\big(-1+6\cos^4(\mb c)\big){\eps_{ij}}^k\mathring{e}^i_a\mathring{e}^j_b.
\ee
\end{subequations}

\section{Heuristic effective dynamics}
\label{sec:4}

\noindent Collecting the above results, i.e. \eqref{harmonic gauge gravitational hamiltonian} and \eqref{curvature for dimension d}, we get that the regularized total Hamiltonian in the harmonic gauge is given by
\be\label{harmonic hamiltonian arbitrary d}
\d\,\mathcal{C}^\text{h}=\d\,\mathcal{C}_\text{g}^\text{h}+\mathcal{C}_\text{m}^\text{h}={}^{(2)}\mathcal{C}_\text{g}^\text{h}\mathcal{F}_d(\theta)+\f{1}{2}p_\phi^2=-\f{3}{\kappa\gamma^2\mb^2}|p|^2\sin^2(\mb c)\mathcal{F}_d(\theta)+\f{1}{2}p_\phi^2,
\ee
where, to facilitate the comparison between the case $d=2$ and the case of arbitrary $d>2$, we have introduced
\be
\mathcal{F}_d(\theta)\coloneqq-\f{3}{d(d^2-1)}\f{1}{\sin\theta}\f{\partial}{\partial\theta}\left(\f{\sin(d\theta)}{\sin\theta}\right).
\ee
The expression \eqref{harmonic hamiltonian arbitrary d} now has to be turned into a quantum operator in order to study the quantum dynamics (without inverse-volume corrections). As can be anticipated, the effect of higher spin representations in the quantum theory will come from (the operator version of) the function $\mathcal{F}_d(\theta)$. Some properties of this function are given in appendix \ref{appendixB}, where it is shown in particular that it is a polynomial in $\cos\theta$. By using the explicit expression of the class angle, one therefore gets that $\mathcal{F}_d(\theta)$ is a polynomial in $\cos(\mb c)$ and/or $\sin(\mb c)$ depending on the choice of regularization for the curvature.

Before moving on to the construction of the quantum theory, we would like to briefly study the classical dynamics obtained by using \eqref{harmonic hamiltonian arbitrary d} as a heuristic\footnote{We use here the term ``heuristic effective Hamiltonian'' (as opposed to ``rigorous effective Hamiltonian'') in order to stress the fact that this effective Hamiltonian has not been derived from the semi-classical limit of the quantum theory, but is simply a classical ansatz.} effective classical Hamiltonian. This will enable us to understand further certain of the properties of $\mathcal{F}_d(\theta)$.

\subsection{Equations of motion}

\noindent If we consider \eqref{harmonic hamiltonian arbitrary d} a as heuristic effective classical constraint, we can go to the proper time gauge without worrying about inverse-volume corrections. Then, the evolution in proper time is given as usual by $\partial_t(\cdot)=\lb\cdot,\d\,\mathcal{C}^1\rb$. Furthermore, it is now convenient to switch to the canonical variables \eqref{b and v variables}, in terms of which we get that
\be
\mb c=\lambda b,\q\q|p|^{3/2}=2\pi\gamma\lp^2|\nu|=\f{1}{2\sqrt{3}}\lambda^2|\nu|,
\ee
where $\lambda$ is defined in \eqref{lambda}. With these variables, the total Hamiltonian in the proper time gauge becomes
\be\label{N=1 hamiltonian}
\d\,\mathcal{C}^1=-\f{\sqrt{3}}{2\kappa\gamma^2}|\nu|\sin^2(\lambda b)\mathcal{F}_d(\theta)+\f{\sqrt{3}}{\lambda^2}|\nu|^{-1}p_\phi^2,
\ee
and we get the following set of equations of motion:
\begin{subequations}
\be
\dot{b}&=-\f{\sqrt{3}}{\kappa\gamma^2\hbar}\sin^2(\lambda b)\mathcal{F}_d(\theta)-\f{2\sqrt{3}}{\lambda^2\hbar}|\nu|^{-2}p_\phi^2=-\f{2\sqrt{3}}{\kappa\gamma^2\hbar}\sin^2(\lambda b)\mathcal{F}_d(\theta),\label{b dot eom}\\
\dot{\nu}&=\f{\sqrt{3}}{\kappa\gamma^2\hbar}|\nu|\sin^2(\lambda b)\left(\lambda\cot(\lambda b)\mathcal{F}_d(\theta)+\f{\partial\theta}{\partial b}\f{\partial}{\partial\theta}\mathcal{F}_d(\theta)\right),\\
\dot{\phi}&=\f{2\sqrt{3}}{\lambda^2}|\nu|^{-1}p_\phi,\\
\dot{p}_\phi&=0.
\ee
\end{subequations}
To obtain the second equality in \eqref{b dot eom}, we have used the vanishing of the total Hamiltonian to write
\be\label{constraint surface}
p_\phi^2=\f{\lambda^2}{2\kappa\gamma^2}|\nu|^2\sin^2(\lambda b)\mathcal{F}_d(\theta).
\ee

We can now follow the standard strategy of using the scalar field as an internal clock, and compute the derivative\footnote{Alternatively, we can arrive at this equation by computing the Poisson bracket $b'=\lb b,p_\phi\rb$, where $p_\phi$ is seen as the deparametrized Hamiltonian given by the square root of \eqref{constraint surface}.}
\be\label{first order b'}
b'\coloneqq\f{\partial b}{\partial\phi}=\f{\partial b}{\partial t}\f{\partial t}{\partial\phi}=\dot{b}(\dot{\phi})^{-1}=\dot{b}\f{\lambda^2}{2\sqrt{3}}|\nu|p_\phi^{-1}=\mp\sqrt{\f{2}{\kappa}}\f{\lambda}{\gamma\hbar}\sin(\lambda b)\sqrt{\mathcal{F}_d(\theta)},
\ee
where in the last step we have used the $\pm$ square root of \eqref{constraint surface}. As we will see below, the function $\mathcal{F}_d(\theta)$ has $d-2$ zeros and takes negatives values on $[0,\pi]\ni\theta$. Therefore, instead of evaluating numerically the above first order differential equation for $b$ as a function of the internal scalar field time $\phi$, it is more convenient to solve the following second order differential equation:
\be
b''=\f{2\lambda^2}{\kappa\gamma^2\hbar^2}\sin(\lambda b)\left(\lambda\cos(\lambda b)\mathcal{F}_d(\theta)+\f{1}{2}\sin(\lambda b)\f{\partial\theta}{\partial b}\f{\partial}{\partial\theta}\mathcal{F}_d(\theta)\right),
\ee
where the initial condition for $b'(0)$ is given by plugging $b(0)$ in \eqref{first order b'}. Finally, once the solution $b(\phi)$ has been found, one can use the vanishing of the Hamiltonian constraint to find the following evolution equation for the matter density:
\be\label{density as function of b}
\rho=\f{3}{\kappa\gamma^2\lambda^2}\sin^2(\lambda b)\mathcal{F}_d(\theta),
\ee
where $\theta$ has to be seen as a function $\theta\big(b(\phi)\big)$.

The heuristic effective dynamics for $b(\phi)$ and $\rho(\phi)$ depends obviously on the choice of regularization through the class angle $\theta$. In order to treat the two regularization schemes HR and CR at once, let us denote the corresponding class angles by
\be
\theta_1\coloneqq\arccos\left(\cos(\lambda b)+\f{1}{2}\sin^2(\lambda b)\right),\q\q\theta_2\coloneqq\arccos\big(\cos^2(\lambda b)\big).
\ee
With this, we can now go ahead and solve the above equations.

For $d=2$, we have $\mathcal{F}_2(\theta)=1$, and we recover of course the standard results of effective LQC. As can be seen on figure \ref{fig:1} (we perform all numerical computations with $G=\hbar=1$), the Universe is contracting for $\phi\in(-\infty,0)$, with its density increasing, and expanding for $\phi\in(0,+\infty)$, with its density decreasing. At $\phi=0$, with our choice $b(0)=\pi/(2\lambda)$ of initial conditions, the density reaches the upper bound $\rho_\text{c}$ and the contracting and expanding branches are connected by the quantum bounce. During the evolution, the canonical variable $b$ evolves from $\lim_{\phi\rightarrow-\infty}b(\phi)=0$ to $\lim_{\phi\rightarrow+\infty}b(\phi)=\pi/\lambda$.
\begin{figure}[h]
\begin{center}
\includegraphics{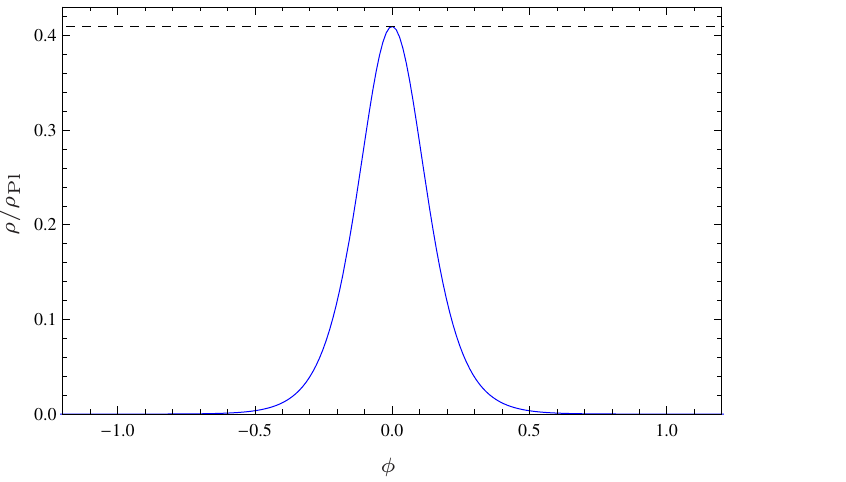}
\hspace{-1cm}
\includegraphics{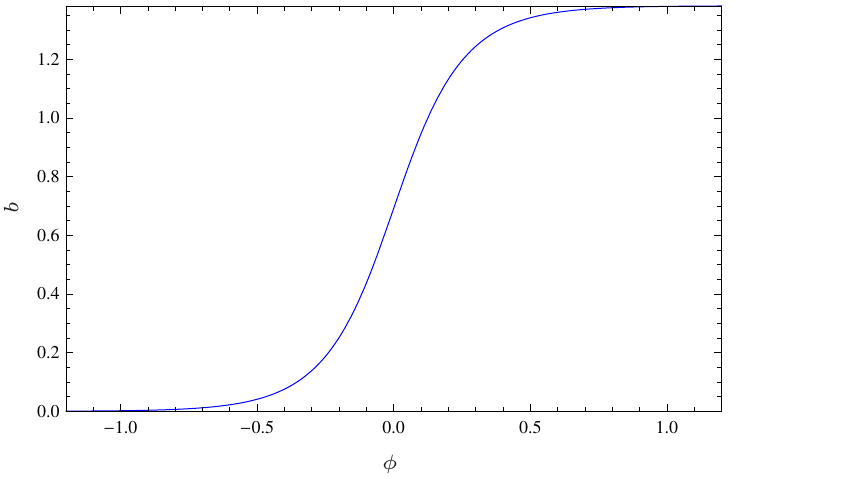}
\caption{Evolution, for $d=2$, of the energy density $\rho$ in Planck units (left) and the canonical variable $b$ (right) as functions of the scalar field time $\phi$. The dotted horizontal line lies at the critical density $\rho_\text{c}/\rho_\text{Pl}$.}
\label{fig:1}
\end{center}
\end{figure}

For $d=3$, we have $\mathcal{F}_3(\theta)=\cos\theta$. Using the regularization scheme CR and therefore the class angle $\theta=\theta_2$, we obtain an effective dynamics which is different from that of standard $d=2$ LQC. As can be seen on figure \ref{fig:2} (using initial conditions $b(0)=\pi/(2\lambda)$), the energy density is still bounded from above, indicating that the evolution is once again singularity-free. However, one can now observe that the density $\rho$ reaches its maximum $\rho_\text{max}$ more than once, and that this critical value of the density is lower than the critical density \eqref{critical density} arising for $d=2$. We are going to see below that this is a generic feature of the choice $d>2$.
\begin{figure}[h]
\begin{center}
\includegraphics{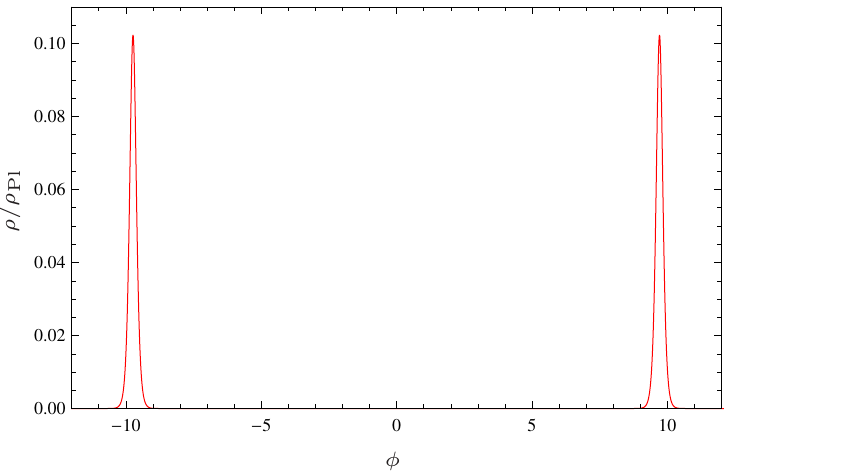}
\hspace{-1cm}
\includegraphics{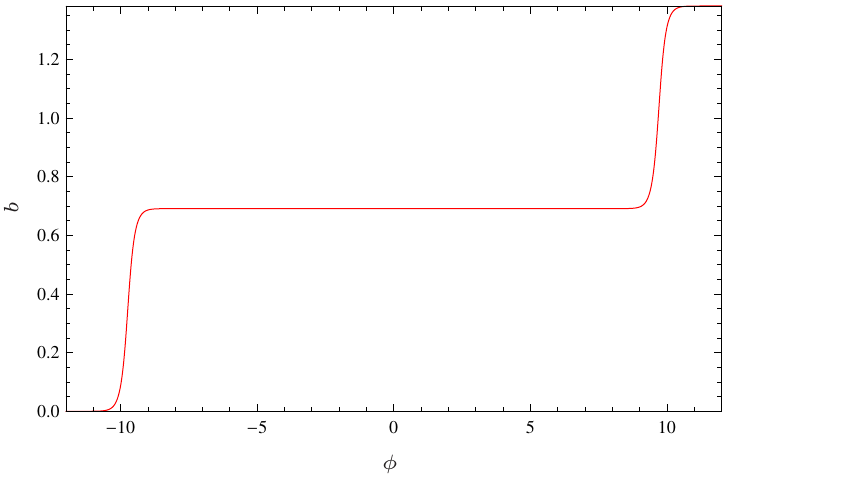}
\caption{Evolution, for $d=3$ and $\theta=\theta_2$, of the energy density $\rho$ in Planck units (left) and the canonical variable $b$ (right) as functions of the scalar field time $\phi$.}
\label{fig:2}
\end{center}
\end{figure}

\subsection{Boundedness of the density}

\noindent In order to understand further how $\mathcal{F}_d(\theta)$ affects the evolution of the density, let us study this function in more details. Remember that the class angle is a function $\theta=\theta(b)$ of the canonical variable $b$, which takes values in the compact interval $[0,\pi/\lambda]\ni b$. On this interval, the arguments of the class angles $\theta_1$ and $\theta_2$ satisfy
\be
-1\leq\cos(\lambda b)+\f{1}{2}\sin^2(\lambda b)\leq1,\q\q0\leq\cos^2(\lambda b)\leq1,
\ee
and we therefore always have that
\be
0\leq\theta_1\leq\pi,\q\q0\leq\theta_2\leq\pi/2.
\ee
The angles $\theta_1$ and $\theta_2$ as functions of $\lambda b\in[0,\pi]$ are represented on figure \ref{fig:theta} below.
\begin{figure}[h]
\begin{center}
\includegraphics{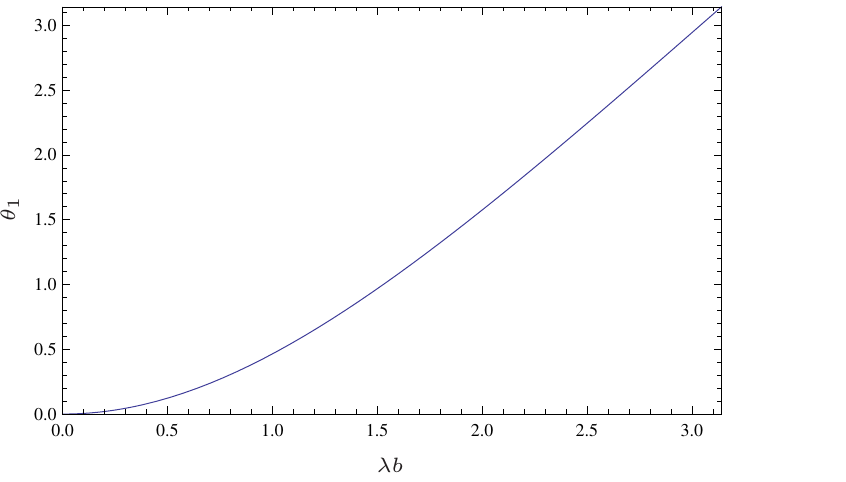}
\hspace{-1cm}
\includegraphics{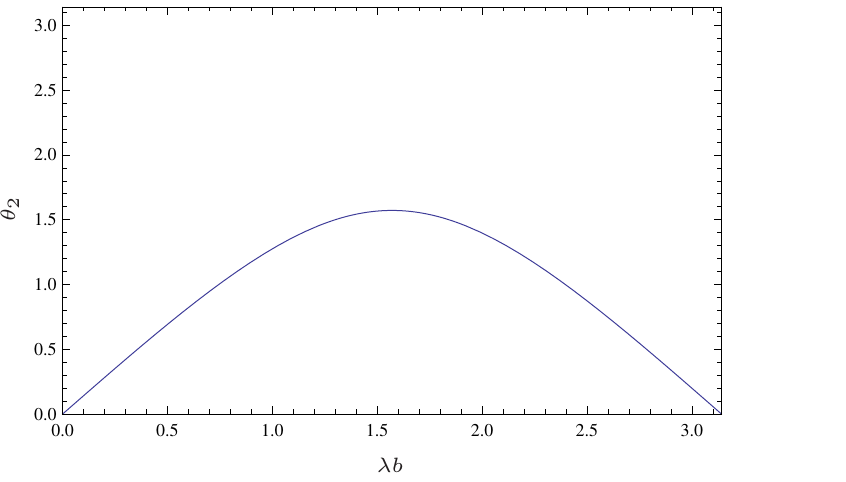}
\caption{Class angles $\theta_1$ (left) and $\theta_2$ (right) for the regularization schemes HR and CR, as functions of $\lambda b$.}
\label{fig:theta}
\end{center}
\end{figure}

By studying $\mathcal{F}_d(\theta)$ on $[0,\pi]\ni\theta$, and using the relation \eqref{density as function of b}, we can therefore infer how the matter density evolves when $b$ spans $[0,\pi/\lambda]$. As can be seen by using the explicit polynomial form \eqref{theta polynomial}, the function $\mathcal{F}_d(\theta)$ is continuous on $[0,\pi]$ (it is actually continuous on $\mathbb{R}$), satisfies
\be
\mathcal{F}_d(0)=1,\q\q\mathcal{F}_d(\pi/2)=
\left\{\begin{array}{l}
0\phantom{\displaystyle3(-1)^{d/2-1}}\text{for $d$ odd},\\[8pt]
\displaystyle\f{3(-1)^{d/2-1}}{d^2-1}\phantom{0}\text{for $d$ even},
\end{array}\right.
\q\q\mathcal{F}_d(\pi)=(-1)^{d},
\ee
for all $\mathbb{N}\ni d\geq2$, and $\mathcal{F}_2(\theta)=1$ for all $\theta$. Since we are dealing with a continuous function on a bounded interval, we are guaranteed that the density $\rho$ is never divergent throughout the effective classical evolution. However, we still have to make sure that, as $d$ increases, the image of $\mathcal{F}_d(\theta)$ remains bounded, thereby ensuring that the critical density itself does not become arbitrarily large when $d$ increases. As can be seen on the bound \eqref{bound on F}, this condition is indeed satisfied.

We can therefore conclude that, given the effective classical evolution equation \eqref{density as function of b} for $\rho$ as a function of $b\in[0,\pi/\lambda]$, we have $\rho\leq\rho_\text{c}$ for any choice of fixed $d\geq2$. Moreover, this holds regardless of the regularization scheme which is used for the curvature. As an illustration, the plot of the function $\mathcal{F}_d(\theta)$ for some values of $d$ is given in figure \ref{fig:3}.
\begin{figure}[h]
\begin{center}
\includegraphics{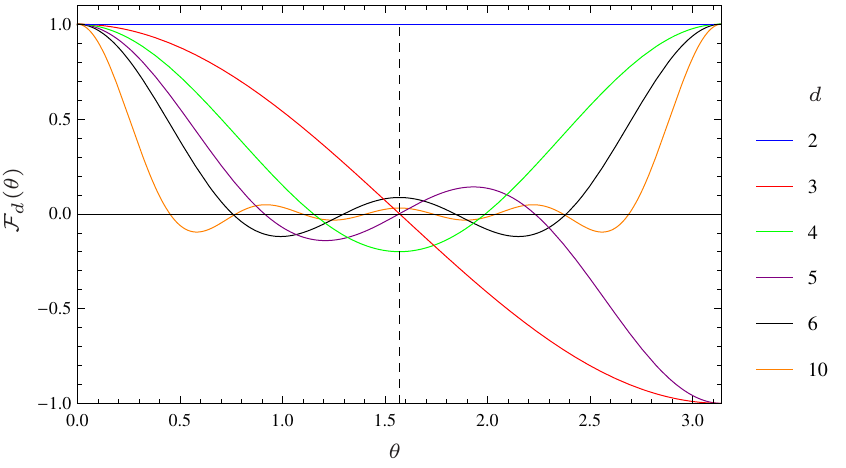}
\caption{Plot of the function $\mathcal{F}_d(\theta)$ for $\theta\in[0,\pi]$ and for different values of $d$ listed on the right. The dotted vertical line is at $\theta=\pi/2$.}
\label{fig:3}
\end{center}
\end{figure}

\subsection{Evolution of the density with a single-spin Hamiltonian}

\noindent As we have seen in the previous subsection, for any $d>2$ the function $\mathcal{F}_d(\theta)$ changes sign during the evolution of the variable $b$ on $[0,\pi/\lambda]$. As a result, the density obtained from the heuristic classical equation \eqref{density as function of b} has the unlikeable feature of not remaining positive. This can be seen in more details on figure \ref{fig:4}, where we have represented $\rho/\rho_\text{Pl}$ as a function of $b$ for different values of $d$ and for the two regularization schemes HR and CR. Note that equation \eqref{density as function of b}, which is responsible for the change of sign of $\rho$ because of the presence of $\mathcal{F}_d(\theta)$, is obviously in conflict with the standard relation \eqref{classical density}, which is manifestly positive. This is simply because we have here modified the dynamics by considering the effective LQC Hamiltonian constraint with arbitrary $d$.

This unsettling appearance of negative values of the density does most likely not reflect the emergence of new exotic physics when $d>2$, but rather the fact that the heuristic constraint \eqref{harmonic hamiltonian arbitrary d} is not a viable effective classical Hamiltonian. This is in fact to be expected since we have not derived \eqref{harmonic hamiltonian arbitrary d} from the proper semi-classical limit of the quantum theory. From this point of view, the fact that for $d=2$ the proper effective semi-classical Hamiltonian happens to coincide with the classical regularized Hamiltonian with holonomy corrections appears as a coincidence. 

In addition to this, one can see that the two different regularization schemes for the curvature operator lead to different evolution profiles for the density. Indeed, the bounce is not symmetric (in the contracting and expanding branches) for $\theta=\theta_1$ (i.e. for HR), while it is symmetric for $\theta=\theta_2$ (i.e. for CR).

\begin{figure}[h]
\begin{center}
\includegraphics{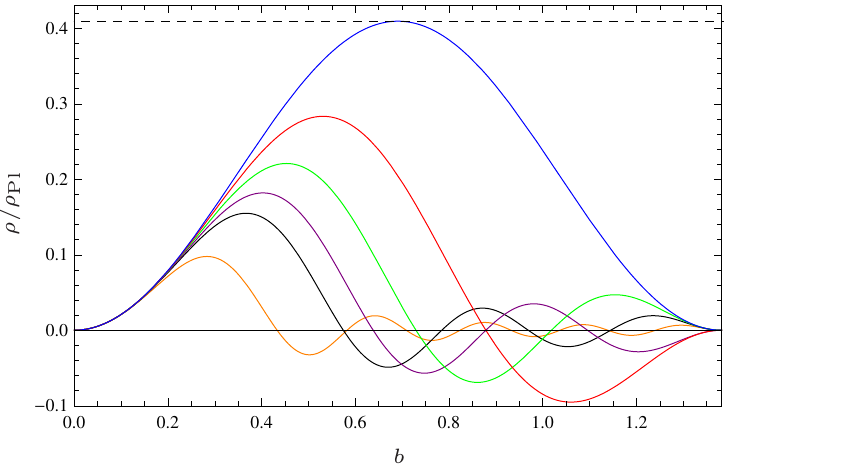}
\hspace{-1cm}
\includegraphics{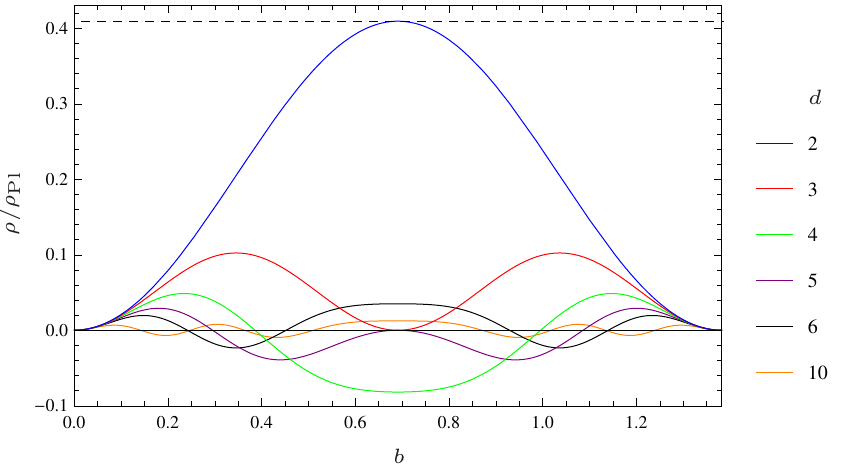}
\caption{Plots of the energy density $\rho$ (in Planck units) for $b\in[0,\pi/\lambda]$ and for different values of $d$ listed on the right. The dotted horizontal line represents the value $\rho=\rho_\text{c}$. The plot on the left is for $\theta=\theta_1$ and the one on the right for $\theta=\theta_2$.}
\label{fig:4}
\end{center}
\end{figure}

\subsection{Evolution of the density with a multiple-spin Hamiltonian}

\noindent Instead of considering the heuristic effective classical dynamics which arises for a single choice of $\SU(2)$ representation, one can consider a more complicated Hamiltonian which contains a contribution from several spins. As we are going to see, this leads generically to a heuristic effective classical dynamics for which the density remains positive throughout the evolution.

Let us consider a heuristic effective gravitational Hamiltonian which is given by a linear combination of $n$ constraints regularized with different values of spins. One reasonable possibility is to choose these spins to be successive half-integers from $j=1/2$ to $j=(n-1)/2$. In this case, the matter energy density can be written as
\be\label{summed rho}
\rho=\f{3}{\kappa\gamma^2\lambda^2}\sin^2(\lambda b)\sum_{d=2}^n\alpha_d\mathcal{F}_d(\theta),
\ee
with weights satisfying
\be
\alpha_d>0,\q\q\sum_{d=2}^n\alpha_d=1.
\ee
For the sake of definiteness, in order to satisfy these conditions let us further write the weights $\alpha_d$ in the form
\be
\alpha_d=\beta_d\left(\sum_{i=2}^n\beta_i\right)^{-1},
\ee
where $\beta_d>0$. In the absence of a more thorough connection between LQC and the full theory, it is not clear what physical principle should determine the choice of the weights $\beta_d$. For example, we could choose these weights to be the same for all the different spin values i.e. $\beta_d=1$, or assign higher weights to the lower spins, using for example an exponential damping for the contribution of the higher dimensions $d$.

The energy density as a function of $b$ is represented on figure \ref{fig:5} for three different choices of weights. The two upper plots are for exponentially decreasing weights, where we have chosen $\beta_d=\exp(-\sqrt{d})$. The two middle plots are for linearly decreasing weights, with $\beta_d=n+1-d$. The two lower plots are for equal weights, which is obtained by setting $\beta_d=1$. As announced above, one can see that the density remains positive throughout the evolution, and that the maximal density $\rho_\text{max}$ reached during the evolution is a decreasing function of the number $n$ of representations appearing in the Hamiltonian. As expected, we also recover the fact that the evolution is asymmetric for HR while it is symmetric for CR. When comparing the results obtained within a same regularization scheme HR or CR, one can see that there is a priori no major difference between the three choices of weights.

\newpage
\begin{figure}[h]
\begin{center}
\includegraphics{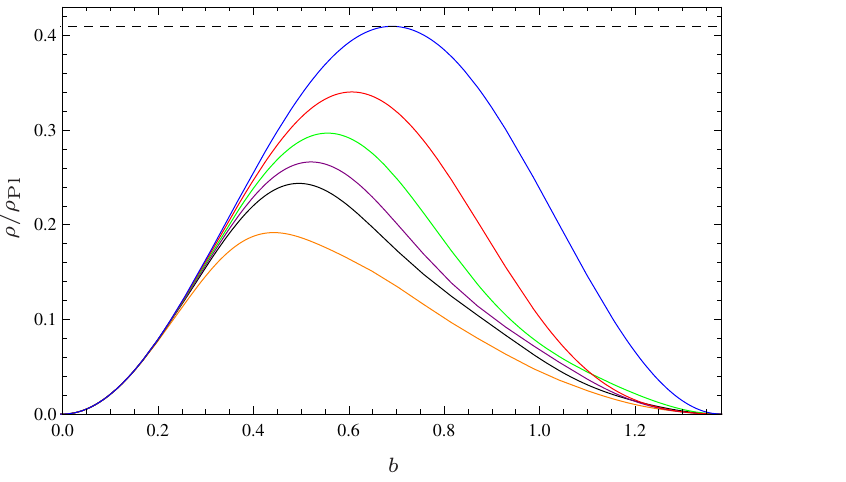}
\hspace{-1cm}
\includegraphics{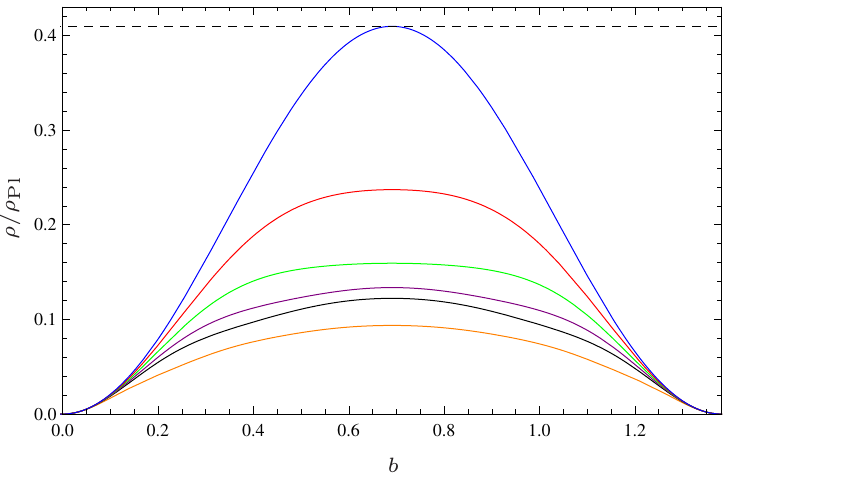}
\includegraphics{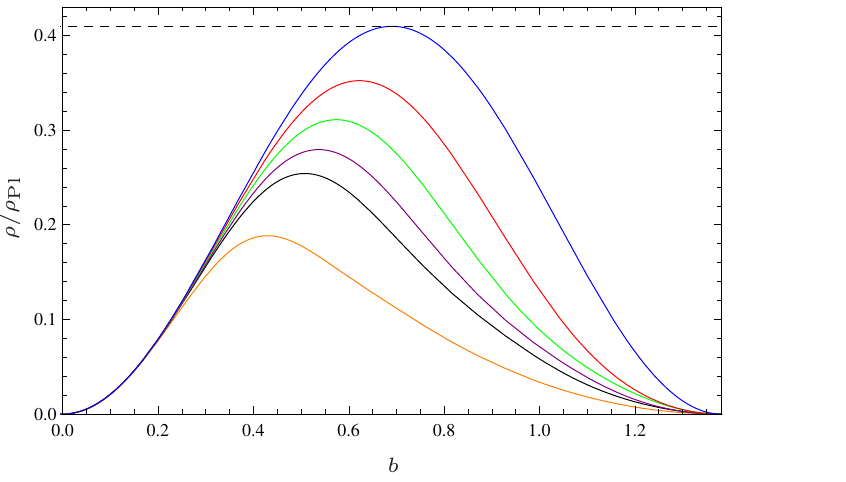}
\hspace{-1cm}
\includegraphics{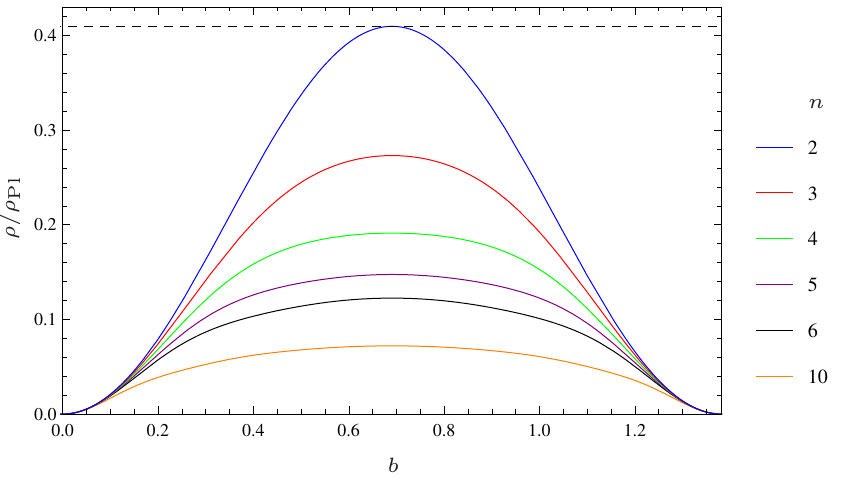}
\includegraphics{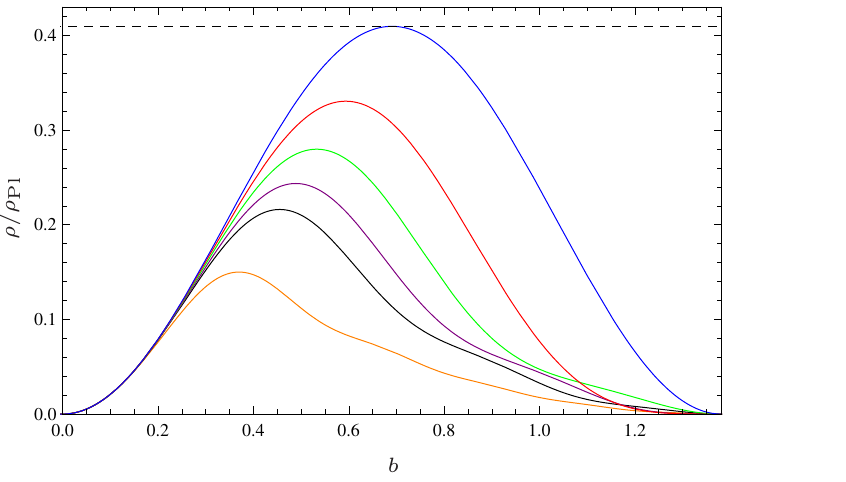}
\hspace{-1cm}
\includegraphics{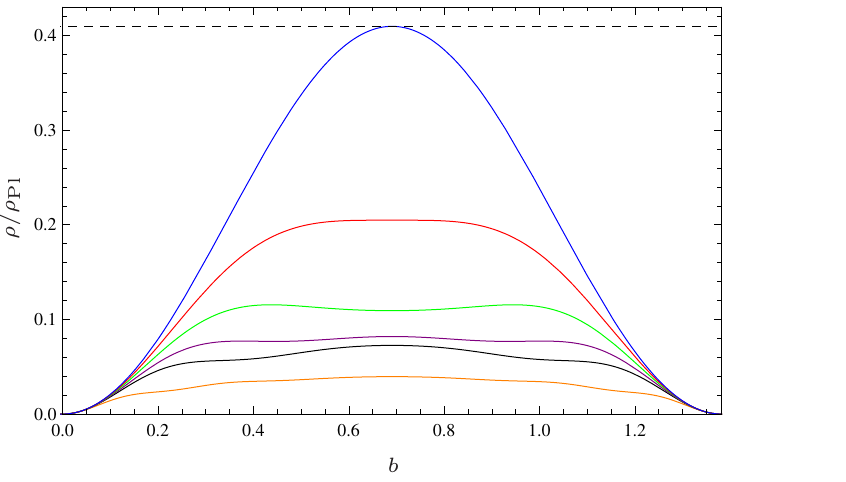}
\caption{Plots of the energy density $\rho$ (in Planck units) for $b\in[0,\pi/\lambda]$ and for different values of $n$ listed on the right. The left column is for $\theta=\theta_1$, while the one on the right is for $\theta=\theta_2$. The two upper plots are for exponentially decreasing weights with $\beta_d=\exp(-\sqrt{d})$. The two middle plots are for linearly decreasing weights with $\beta_d=n+1-d$. The two lower plots are for equal weights with $\beta_d=1$.}
\label{fig:5}
\end{center}
\end{figure}

However, if one increases the number $n$, which represents the maximal dimension allowed in the sum \eqref{summed rho}, the evolution of the density starts to depend strongly on the choice of weights $\beta_d$. More precisely, the maximal value $\rho_\text{max}$ reached along the evolution converges to zero if the weights $\beta_d$ are not decreasing fast enough with the dimension $d$. This can be seen on figure \ref{fig:6}, where we have represented the maximal density as a function of $n$ for different choices of weights. For equal weights $\beta_d=1$ or linearly decreasing weights $\beta_d=n+1-d$, the maximal density drops down to zero. On the other hand, for exponentially decreasing weights the density converges to a value $\bar{\rho}_\text{max}<\rho_\text{c}$. Numerically, for $n=2000$ one finds the values of $\bar{\rho}_\text{max}$ which are reported in table \ref{table:1}.

\begin{figure}[h]
\begin{center}
\hspace{-1cm}
\includegraphics{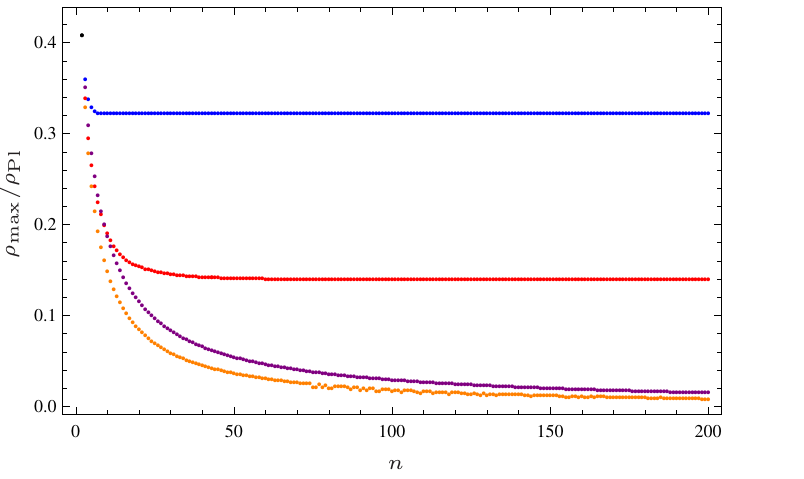}
\hspace{-1cm}
\includegraphics{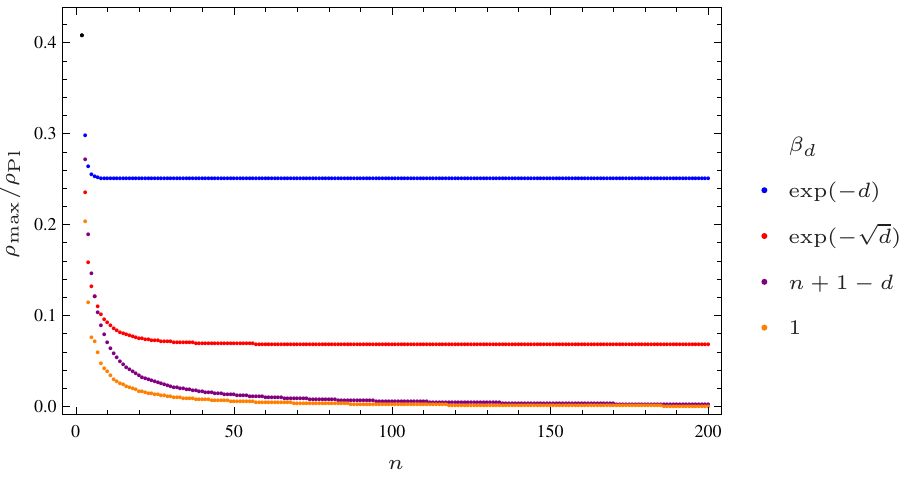}
\caption{Plot of the maximal density $\rho_\text{max}$ (in Planck units) as a function of the largest dimension $n$ allowed in the summed expression \eqref{summed rho}, for different choices of weights $\beta_d$, and for the regularization HR (left) and CR (right). The black dot in the top left corner corresponds to $n=2$ and the standard LQC value $\rho_\text{max}=\rho_\text{c}$. Note that the slight fluctuations appearing after $n>70$ for $\beta_d=1$ in the scheme HR (left) are simply due to numerical uncertainties.}
\label{fig:6}
\end{center}
\end{figure}

\begin{table}[h]
\begin{center}
\begin{tabular}{|c|c|c|}
\hline
\backslashbox{~~weight $\beta_d$}{scheme~~} & ~~HR~~ & ~~CR~~ \\
\hline
$\exp(-d)$ & ~~0.32317~~ & ~~0.25215~~ \\
\hline
$\exp(-\sqrt{d})$ & 0.14078 & 0.06961 \\
\hline
$n+1-d$ & 0.00183 & 0.00035 \\
\hline
$1$ & 0.00096 & 0.00018 \\
\hline
\end{tabular}
\end{center}
\caption{Maximal density $\bar{\rho}_\text{max}/\rho_\text{Pl}$ attained for $n=2000$, for different choices of weights and for the two regularization schemes.}
\label{table:1}
\end{table}

\section{Quantum theory}
\label{sec:5}

\noindent In this section, we briefly comment along the lines of \cite{V,BD} on the (in)stability of the solutions to the difference equations defining the quantum dynamics, and on the absence of relationship between this (in)stability and the change of sign in the energy density which has been observed in the previous sections.

Using the conjugate variables $(b,\nu)$, the total Hamiltonian \eqref{harmonic hamiltonian arbitrary d} in the harmonic gauge becomes
\be
\d\,\mathcal{C}^\text{h}=-\f{3\pi G\hbar^2}{2\lambda^2}|\nu|^2\sin^2(\lambda b)\mathcal{F}_d(\theta)+\f{1}{2}p_\phi^2.
\ee
This is the expression which has to be turned into a quantum operator. Let us work with wavefunctions in a choice of polarization $\Psi(\nu,\phi)$. On these wavefunctions, the exponentiated connection operators act as finite translations, i.e. as
\be
\widehat{\exp(\pm\i\mb c)}\Psi(\nu,\phi)=\widehat{\exp(\pm\i\lambda b)}\Psi(\nu,\phi)=\Psi(\nu\pm2\lambda,\phi),
\ee
while the volume operator acts by multiplication, i.e. as
\be
\widehat{V}\Psi(\nu,\phi)=2\pi\gamma\lp^2|\hat{\nu}|\Psi(\nu,\phi)=2\pi\gamma\lp^2|\nu|\Psi(\nu,\phi).
\ee
Using as usual a Schr\"odinger representation for the massless scalar field representing the matter sector, we get the following quantum evolution equation:
\be
\partial^2_\phi\Psi(\nu,\phi)=-\f{3\pi G}{\lambda^2}|\hat{\nu}|^2\widehat{\sin^2(\lambda b)}\widehat{\mathcal{F}}_d(\theta)\Psi(\nu,\phi).
\ee
Here our choice of operator ordering is such that the volume operator $|\hat{\nu}|^2$ is to the left. This is the choice which will lead to the simplest expression for the difference equations when the $b$-dependent operators will act on the wavefunctions.

By studying the explicit expression for $\mathcal{F}_d(\theta)$ as a function of $\cos(\lambda b)$ and/or $\sin(\lambda b)$, one can see that, regardless of the regularization scheme which is used for the curvature, the right-hand side of the quantum evolution equation will always be a difference equation of step $\Delta\lambda=4(d-1)=8j$. In order to see this explicitly on a few examples, let us introduce the notations $\Psi_{\pm n}\coloneqq\Psi(\nu\pm n\lambda,\phi)$ and $\Psi\coloneqq\Psi_0=\Psi(\nu,\phi)$.
For $d=2$ we get
\be\label{d=2 difference equation}
\partial^2_\phi\Psi=\f{3\pi G}{\lambda^2}\nu^2\f{1}{4}\big(\Psi_{+4}-2\Psi+\Psi_{-4}\big).
\ee
For $d=3$ and $\theta=\theta_1$ we get
\be
\partial^2_\phi\Psi=-\f{3\pi G}{\lambda^2}\nu^2\f{1}{32}\big(\Psi_{+8}-4\Psi_{+6}-4\Psi_{+4}+4\Psi_{+2}+6\Psi+4\Psi_{-2}-4\Psi_{-4}-4\Psi_{-6}+\Psi_{-8}\big).
\ee
For $d=3$ and $\theta=\theta_2$ we get
\be\label{d=3 difference equation}
\partial^2_\phi\Psi=\f{3\pi G}{\lambda^2}\nu^2\f{1}{16}\big(\Psi_{+8}-2\Psi+\Psi_{-8}\big).
\ee
For $d=4$ and $\theta=\theta_1$ we get
\be
\partial^2_\phi\Psi&=\f{3\pi G}{\lambda^2}\nu^2\f{1}{640}\big(3\Psi_{+12}-24\Psi_{+10}+30\Psi_{+8}+72\Psi_{+6}+13\Psi_{+4}-48\Psi_{+2}-92\Psi\nn\\
&\phantom{=\f{3\pi G}{\lambda^2}\nu^2\f{1}{640}\big(}-48\Psi_{-2}+13\Psi_{-4}+72\Psi_{-6}+30\Psi_{-8}-24\Psi_{-10}+3\Psi_{-12}\big).
\ee
For $d=4$ and $\theta=\theta_2$ we get
\be
\partial^2_\phi\Psi=\f{3\pi G}{\lambda^2}\nu^2\f{1}{160}\big(3\Psi_{+12}+6\Psi_{+8}-11\Psi_{+4}+4\Psi-11\Psi_{-4}+6\Psi_{-8}+3\Psi_{-12}\big).
\ee

In order to analyze the stability of these difference equations, we will adopt the procedure developed in \cite{BD}. In the large volume limit, one can look for solutions for the wavefunction of the form $\Psi_{\pm n}(z)=\Psi(\nu\pm n\lambda)=z^{\pm n}$, and demand the vanishing of the right-hand side of the quantum evolution equation. This leads to a homogeneous equation in $z$. The requirement of stability is that this equation has no solutions with norm greater than one. Let us apply this criterion to the case $d=2$. The difference equation and the related homogeneous equation are given by
\be
\Psi_{+4}-2\Psi+\Psi_{-4}=0\q\rightarrow\q z^4-2+z^{-4}=0.
\ee
The solutions are $z\in\{-1,1,i,-i\}$ and all satisfy $|z|^2=1$. This difference equation is therefore stable and there are no spurious solutions spoiling the semi-classical limit (i.e. the large volume behavior).

We can now turn to the case $d=3$. Working with $\theta_1$ (and therefore HC) as in \cite{V}, the difference equation leads to the homogeneous equation
\be
z^8-4z^6-4z^4+4z^2+6+4z^{-2}-4z^{-4}-4z^{-6}+z^{-8}=0,
\ee
which does admit solutions such that $|z|^2>1$. These spurious solutions turn out to spoil the semi-classical limit of the theory, as already pointed out in \cite{V}, making this particular realization of the $d=3$ model of LQC unstable. However, we can proceed to the same stability analysis for the difference equation obtained with $\theta_2$ (and therefore from CR). This leads to
\be
z^8-2+z^{-8}=0,
\ee
and one can check that all the solutions to this equation satisfy $|z|^2=1$. Using the freedom in choosing a different regularization scheme for the curvature (namely CR instead of HR), it is therefore possible, for flat FLRW models at least, to obtain a stable quantum theory for $d=3$. It is straightforward to see that \eqref{d=2 difference equation} and \eqref{d=3 difference equation} can be mapped to each other by simply redefining the step size $\lambda$. This has the simple consequence that the quantum theories, in these two cases, have the exact same qualitative behavior, with only minor differences in some numerical prefactors. For example, the critical density for $(d=3,\text{CR})$ will be lower than that for $d=2$, but all other conceptual features will remain exactly identical. It is probably wise to interpret this ``standard'' behavior of the $(d=3,\text{CR})$ LQC model as a coincidence due to the existence of the regularization CR. It should be clear from the analysis of this section that other choices of $d$ will lead to completely different physics in the quantum theory, and that this latter should be defined (if at all possible) with extra care.

Now, one can also check numerically that any choice $d>4$ will fail the stability criterion of \cite{BD}. As explained in the introduction, it is tempting at this point to relate the failure of passing this stability criterion, i.e. the existence of spurious solutions, to the fact that the heuristic effective equations at the classical level lead to a negative energy density. Indeed, for $d=2$ the density remains positive and the stability criterion is passed. For $(d=3,\text{CR})$, one can see on figure \ref{fig:4} that the density remains positive, and we have just shown that the stability criterion is passed as well. However, for $(d=3,\text{HR})$ and any $d>3$, the density will reach negative values during the Hubble evolution, and the stability criterion always fails. These two notions (the sign of the energy density and the stability) therefore seem to be related. However, one rather puzzling feature is that, when using a sum over several representations, the quantum theory still fails the stability criterion although the energy density is positive (as we saw on figure \ref{fig:5}).

To be more concrete, let us consider the quantum evolution equation which is obtained by allowing different representations. Following what we have done in this section and in section \ref{sec:4}, this leads to
\be
\partial^2_\phi\Psi(\nu,\phi)=-\f{3\pi G}{\lambda^2}|\hat{\nu}|^2\widehat{\sin^2(\lambda b)}\sum_{d=2}^n\alpha_d\widehat{\mathcal{F}}_d(\theta)\Psi(\nu,\phi).
\ee
Choosing for example $n=3$ (i.e. the representations $d=2$ and $d=3$), the regularization HR, and the equal weights $\beta_d=1$, we obtain the quantum evolution equation
\be
\partial^2_\phi\Psi=-\f{3\pi G}{\lambda^2}\nu^2\f{1}{64}\big(\Psi_{+8}-4\Psi_{+6}-12\Psi_{+4}+4\Psi_{+2}+22\Psi+4\Psi_{-2}-12\Psi_{-4}-4\Psi_{-6}+\Psi_{-8}\big).
\ee
Classically, the evolution of the density in this case is represented by the red curve on the lower left plot of figure \ref{fig:5}. In spite of this evolution being positive, the homogeneous equation obtained from the above quantum evolution equation has solutions with $|z|^2>1$. For $n=3$, the regularization CR, and the equal weights $\beta_d=1$, we obtain
\be
\partial^2_\phi\Psi=\f{3\pi G}{\lambda^2}\nu^2\f{1}{32}\big(\Psi_{+8}+4\Psi_{+4}-10\Psi+4\Psi_{-4}+\Psi_{-8}\big).
\ee
Now, even though we are summing the two cases discussed above which passed the stability criterion, namely $(d=2,\text{CR})$ and $(d=3,\text{CR})$, we obtain a quantum theory which fails the stability criterion.

Although we do not push the study of the quantum theory further, it would be very interesting to understand more precisely the relationship between the sign of the energy density and the stability criterion of \cite{BD}, and to understand how the spurious solutions could be regulated with the proper construction of the physical inner product. It seems that this is a crucial condition if one wishes to define properly LQC with higher spin representations and establish a connection with the full theory.

\section{Inverse-volume corrections}
\label{sec:6}

\noindent Since the manipulation of classical expressions does not commute with the quantization process, one can naturally expect to obtain different quantum theories for different choices of lapse (and of orderings) when the vanishing of the total Hamiltonian is applied as a quantum constraint. In the previous section, we have studied the construction of the quantum theory using the harmonic time gauge, i.e. with the choice of lapse $N=|p|^{3/2}$. However, in order to mimic more closely the features of the full theory, where the choice of the harmonic time gauge is not viable (see footnote 2), one can also choose to work with the proper time gauge, i.e. with $N=1$. This then leads to the appearance of the so-called inverse-volume corrections.

More precisely, when working with the proper time as opposed to the harmonic time, different powers of the volume $V=|p|^{3/2}$ appear in the total Hamiltonian constraint. First, as can be seen on \eqref{matter hamiltonian}, the Hamiltonian for a massless scalar field contains the inverse-volume contribution $V^{-1}=|p|^{-3/2}$. Using Thiemann's trick, this contribution can be rewritten (or rather regularized) classically as
\be\label{classical inverse-volume}
{}^{(d,r)}V^{-1}=\left[-\f{1}{\tau(d)}\f{1}{\kappa\gamma\mb r}\tr_{(d)}\left(\d h_i^{(\mb)}\lb\d h_{i^{-1}}^{(\mb)},V^{2r/3}\rb\d\tau^i\right)\right]^{3/[2(1-r)]}.
\ee
Using the fact that
\be
\d h_i^{(\mb)}\lb\d h_{i^{-1}}^{(\mb)},V^{2r/3}\rb=-\mb\lb c,V^{2r/3}\rb\d\tau_i=-\f{\kappa\gamma\mb r}{3}|p|^{r-1}\d\tau_i,
\ee
together with the normalization
\be
\tr_{(d)}\left(\d\tau_i\d\tau^i\right)=3\tau(d),
\ee
one finds indeed that \eqref{classical inverse-volume} reduces to $|p|^{3/2}$. Without computing the Poisson bracket classically, and simply turning the right-hand side of \eqref{classical inverse-volume} into a quantum operator instead, one obtains an expression for the inverse-volume operator ${}^{(d,r)}\widehat{V}^{-1}$ which depends on the two regularization ambiguities $d$ and $r\in(0,1)$. In the literature, it is customary to use the values $r=1/2$ or $r=3/4$. Note that there exist additional ambiguities in the definition of the inverse-volume operator, as explained in section IV.A of \cite{SWE}, but we do not consider them here.

Next, one has to look at how the choice $N=1$ modifies the gravitational part of the Hamiltonian constraint. The simplest possibility is to consider the first term in \eqref{N=1 hamiltonian}, which is written in terms of classical quantities that all posses well-defined quantum analogues (at the difference with the second term, which should be written using the above inverse-volume formula). However, in order to mimic more closely \eqref{classical inverse-volume} and the full theory, it is common to also use Thiemann's trick to rewrite the gravitational part of the Hamiltonian constraint. More precisely, one can use the scheme HR combined with Thiemann's trick to write
\be\label{i.v. gravitational hamiltonian 1}
\d\,\mathcal{C}_\text{g}^1=\f{1}{\tau(d)}\f{1}{\kappa^2\gamma^3\mb^3}\eps^{ijk}\tr_{(d)}\left(\d h^{(\mb)}_{\Box_{ij}}\d h_k^{(\mb)}\lb\d h_{k^{-1}}^{(\mb)},V\rb\right).
\ee
Alternatively, one can use the scheme CR together with Thiemann's trick to write
\be\label{i.v. gravitational hamiltonian 2}
\d\,\mathcal{C}_\text{g}^1=\f{1}{\tau(d)}\f{1}{32\kappa^2\gamma^3\mb^3}\eps^{ijk}\tr_{(d)}\left(\Big[\d h_i^{(2\mb)}-\d h_{i^{-1}}^{(2\mb)},\d h_j^{(2\mb)}-\d h_{j^{-1}}^{(2\mb)}\Big]\d h_k^{(2\mb)}\lb\d h_{k^{-1}}^{(2\mb)},V\rb\right).
\ee
When expanding the holonomies in powers of $\mb$ and taking the limit $\mb\rightarrow0$, one can indeed check that these two expressions lead to \eqref{reduced gravitational hamiltonian} with $N=1$.

To go to the quantum theory, one should now take either of these expressions, turn the holonomies and the volume into quantum operators, and replace the Poisson bracket by a commutator. However, one then runs into the problem that formula \eqref{trace formula} cannot be used in order to evaluate the trace in an arbitrary representation of dimension $d$. This is because in the quantum theory, once the Poisson bracket has been replaced by a commutator, this trace is not anymore that of a Lie algebra element multiplied by a group element (which would have been the case if the Poisson bracket had been computed classically and the volume factor taken out of the trace). One alternative way of evaluating the traces in \eqref{i.v. gravitational hamiltonian 1} and \eqref{i.v. gravitational hamiltonian 2} is of course to go back to the formulas of \cite{V,CL}, but this would render useless the simplifications which have been obtained in section \ref{sec:3} for the regularization of the curvature. It is therefore convenient at this point to separate the contributions from the regularized curvature and inverse-volume term. Indeed, this latter can be written as
\be
q^{-1/2}E^a_iE^b_j{\eps^{ij}}_k=-\f{1}{\tau(d)}\f{2}{\kappa\gamma\mb\mathring{V}^{1/3}}\tr_{(d)}\left(\d h_l^{(\mb)}\lb\d h_{l^{-1}}^{(\mb)},V\rb\d\tau_k\right)\eps^{ijl}\mathring{e}^a_i\mathring{e}^b_j\mathring{q}^{1/2},
\ee
which, when combined with the expression $\d F^k_{ab}$ for the regularized curvature, leads to the following expression for the gravitational part of the Hamiltonian constraint:
\be\label{i.v. gravitational hamiltonian final}
\d\,\mathcal{C}_\text{g}^1=\f{1}{\tau(d)}\f{2}{\kappa^2\gamma^3\mb^3}\tr_{(d)}\left(\d h_i^{(\mb)}\lb\d h_{i^{-1}}^{(\mb)},V\rb\d\tau^i\right)\sin^2(\mb c)\mathcal{F}_d(\theta).
\ee
This expression is valid for the two regularization schemes HR and CR (since it depends on the choice of angle $\theta$), and one can check that in the limit $\mb\rightarrow0$ it leads back to \eqref{reduced gravitational hamiltonian} with $N=1$.

We are therefore left with the task of evaluating the traces in the quantum versions of \eqref{classical inverse-volume} and \eqref{i.v. gravitational hamiltonian final}, which are obtained by replacing $\lb\cdot\,,\cdot\rb\rightarrow\big[\cdot\,,\cdot\big]/(\i\hbar)$ and turning the holonomies and the volume into quantum operators. Using gauge invariance, we can write
\be
\tr_{(d)}\left(\d\hat{h}_i^{(\mb)}\left[\d\hat{h}_{i^{-1}}^{(\mb)},\widehat{V}^\alpha\right]\d\tau^i\right)=3\tr_{(d)}\left(\d\hat{h}_3^{(\mb)}\left[\d\hat{h}_{3^{-1}}^{(\mb)},\widehat{V}^\alpha\right]\d\tau_3\right),
\ee
where we have left the power of the volume unspecified for generality. Now, this trace can easily be computed using the fact that it involves diagonal matrices. More precisely, for the holonomies we have
\be
\left(\d\hat{h}_3^{(\mb)}\right)_{mn}=\widehat{\exp(\i m\mb c)}\delta_{mn},
\ee
for the commutator term we have
\be
\left(\d\hat{h}_3^{(\mb)}\left[\d\hat{h}_{3^{-1}}^{(\mb)},\widehat{V}^\alpha\right]\right)_{mn}=\left(\widehat{V}^\alpha-\widehat{\exp(\i m\mb c)}\widehat{V}^\alpha\widehat{\exp(-\i m\mb c)}\right)_{mn},
\ee
and for the third generator of $\su(2)$ we have
\be
\left(\d\tau_3\right)_{mn}=\i m\delta_{mn},
\ee
with everywhere $m,n\in\{-j,-j+1,\dots,j-1,j\}$. Putting this together, and switching back to the $(b,\nu)$ representation, we obtain the following expression for the gravitational part of the constraint operator:
\be
\d\,\widehat{\mathcal{C}}_\text{g}^{\,1}=\f{1}{\tau(d)}\f{2\sqrt{3}\pi\lp^2}{\kappa^2\gamma^2\lambda\hbar}|\hat{\nu}|\left(\sum_{m=-j}^jm\left(|\hat{\nu}|-\widehat{\exp(\i m\lambda b)}|\hat{\nu}|\widehat{\exp(-\i m\lambda b)}\right)\right)\widehat{\sin^2(\lambda b)}\widehat{\mathcal{F}}_d(\theta).
\ee
Notice that here we have made a choice of operator ordering similar to that of the previous section, namely with the volume factors on the left. Finally, for the inverse-volume operator we obtain
\be
{}^{(d,r)}\widehat{V}^{-1}=\left[-\f{1}{\tau(d)}\f{\big[3\big(2\pi\gamma\lp^2\big)^r\big]^{2/3}}{\kappa\gamma(2\lambda)^{1/3}\hbar r}|\hat{\nu}|^{1/3}\left(\sum_{m=-j}^jm\left(|\hat{\nu}|^{2r/3}-\widehat{\exp(\i m\lambda b)}|\hat{\nu}|^{2r/3}\widehat{\exp(-\i m\lambda b)}\right)\right)\right]^{3/[2(1-r)]}.
\ee
This operator has a very simple diagonal action on states, ${}^{(d,r)}\widehat{V}^{-1}\Psi(\nu,\phi)={}^{(d,r)}B(\nu)\Psi(\nu,\phi)$, where the eigenvalues are given by
\be
{}^{(d,r)}B(\nu)\coloneqq\left[-\f{1}{\tau(d)}\f{\big[3\big(2\pi\gamma\lp^2\big)^r\big]^{2/3}}{\kappa\gamma(2\lambda)^{1/3}\hbar r}|\nu|^{1/3}\left(\sum_{m=-j}^jm\left(|\nu|^{2r/3}-|\nu+2\lambda m|^{2r/3}\right)\right)\right]^{3/[2(1-r)]}.
\ee
With this, the vanishing of the total Hamiltonian constraint operator finally leads to
\be\label{final formula}
\partial^2_\phi\psi(\nu,\phi)=\f{2}{\hbar^2}\left[{}^{(d,r)}B(\nu)^{-1}\right]\d\,\widehat{\mathcal{C}}_\text{g}^{\,1}\Psi(\nu,\phi).
\ee
This expression is valid for the two regularization schemes HR and CR, and takes into account arbitrary representations as well as the inverse-volume corrections.

\section{Conclusion}

\noindent In this paper, we have pushed further the attempts made in \cite{V,CL,BAGN} to regularize the Hamiltonian constraint of (flat FLRW) LQC with arbitrary spin representations. We have seen that the use of the trace formula \eqref{trace formula} together with its rewriting via \eqref{theta polynomial} provides an extremely efficient way of bypassing the complicated results of \cite{V,CL}, and leads to an expression which can be applied to various regularization schemes for the curvature. As a consequence of this result, we can anticipate that the extension to more complicated cosmological models (including for example spatial curvature or anisotropies) will be rather straightforward. Indeed, we have seen that the extension of any model of LQC from $j=1/2$ to $j>1/2$ involves using the results of the regularization for $j=1/2$ together with the formulas \eqref{trace formula} and \eqref{theta polynomial}. There is therefore a systematic construction for which all the ingredients are now known.

In the LQC literature, lots of attention has been devoted to so-called higher order holonomy corrections \cite{CL,MS,HMS,CL2}. These arise as higher order terms in powers of $\mb$, which have been neglected when dropping the limit $\mb\rightarrow0$ in the expression for the regularized curvature. In \cite{CL,CL2}, it has been suggested that there could be a relationship between these higher order holonomy corrections and the contribution of higher spins representations. Our analysis of the effective dynamics carried out in section \ref{sec:4} seems to indicate that these two type of effects have actually very different consequences. Indeed, in \cite{CL} (see equation (4.44)) it has been shown that the inclusion of higher order holonomy corrections increases the critical density (as the order of the corrections is itself increased), while our analysis shows that the inclusion of more spin representations in the dynamics has the effect of lowering the critical density (as can be seen on figure \ref{fig:6}).

There are three interesting directions in which this work could be extended. First, one should construct rigorously the quantum theory corresponding to higher spin representations, and clarify the status of the spurious solutions which have been identified in \cite{BD} and in section \ref{sec:5}. Second, it would be interesting to see if the higher spin representations could have any phenomenological implications, in the spirit of what has been done in \cite{CL,MS,HMS,CL2} for the higher order holonomy corrections. Finally, we hope that the results presented in this note will help clarify the relationship between LQC and GFT/full LQG.

\begin{center}
\textbf{Acknowledgements}
\end{center}
MG and SB would like to thank Martin Bojowald, Steffen Gielen, and Edward Wilson-Ewing for helpful comments, suggestions, and discussions. SB acknowledges hospitality of Perimeter Institute for Theoretical Physics, where part of this work was completed. MG is supported by Perimeter Institute for Theoretical Physics. Research at Perimeter Institute is supported by the Government of Canada through Industry Canada and by the Province of Ontario through the Ministry of Research and Innovation.

\appendix

\section{Some useful $\boldsymbol{\SU(2)}$ formulae}
\label{appendixA}

\noindent In the representation of spin $j$, the Lie algebra $\su(2)$ is generated by matrices of dimension $d\times d$, where $d=2j+1$ is the dimension of the representation. We will denote these matrices by $\d\tau_i$, and in the case $d=2$ omit the dimension label for the sake of simplicity. Our choice of generators for the fundamental representation is given by the matrices
\be
\tau_1=-\f{\i}{2}
\begin{pmatrix}
~0~&~1~\\
~1~&~0~
\end{pmatrix},\q\q
\tau_2=\f{1}{2}
\begin{pmatrix}
~0~&-1~\\
~1~&~0~
\end{pmatrix},\q\q
\tau_3=-\f{\i}{2}
\begin{pmatrix}
~1~&~0~\\
~0~&-1~
\end{pmatrix},
\ee
which are defined in terms of the Pauli matrices by $2\i\tau_i=\sigma_i$, and satisfy the commutation relations $[\tau_i,\tau_j]={\eps_{ij}}^k\tau_k$ as well as the identity
\be
\tau_i\tau_j=\f{1}{2}{\eps_{ij}}^k\tau_k-\f{1}{4}\delta_{ij}\openone.
\ee
In particular, we have that $-4\tau_i^2=\openone$. From this, one can obtain the following useful identity:
\be
\tau_j\tau_i\tau_j=-\f{1}{4}(2\delta_{ij}-1)\tau_i.
\ee
For arbitrary spin, the representation matrices satisfy
\be
\tr_{(d)}\left(\d\tau_i\d\tau_k\right)=-\f{1}{12}d(d^2-1)\delta_{ik}=-\f{1}{3}j(j+1)(2j+1)\delta_{ik},
\ee
which we use to define the normalization factor
\be\label{trace normalization}
\tau(d)\coloneqq\tr_{(d)}\left(\d\tau_i^2\right)=-\f{1}{12}d(d^2-1).
\ee
The relation between the trace in the fundamental representation and that in an arbitrary representation of dimension $d$ is given by
\be
\tr_{(d)}\left(\d\tau_i\d\tau_k\right)=\f{\tau(d)}{\tau(2)}\tr(\tau_i\tau_k).
\ee

The main technical result which we use the the core of the article in order to derive the regularized expression for the curvature in arbitrary representation is the following:
\begin{proposition}\label{lemma}
Let $\d h$ be an element of $\SU(2)$ and $\d\tau_i$ a generator of $\su(2)$, both taken in the representation of dimension $d$. Then we have that
\be\label{trace formula}
\tr_{(d)}\left(\d h\d\tau_i\right)=-\f{1}{2\sin\theta}\f{\partial}{\partial\theta}\left(\f{\sin(d\theta)}{\sin\theta}\right)\tr(h\tau_i),
\ee
where
\be
\theta=\arccos\left(\f{1}{2}\tr(h)\right)\in[0,\pi]
\ee
is the class angle of the group element $h$ computed in the representation of dimension $d=2$.
\end{proposition}
\begin{proof}
First, it is clear that we have
\be\label{proof 1}
\tr_{(d)}\left(\d h\d\tau_i\right)=\left.\f{\partial}{\partial s}\tr_{(d)}\left(\d h\exp\left(s\d\tau_i\right)\right)\right|_{s=0}.
\ee
Then, the trace on the right-hand side, which is just the trace of an $\SU(2)$ group element, can be expressed using the character formula
\be\label{character formula}
\tr_{(d)}\left(\d h\exp\left(s\d\tau_i\right)\right)=\chi_{(d)}\big(\theta(s)\big)=\f{\sin\big(d\theta(s)\big)}{\sin\theta(s)}.
\ee
Now, the class angle $\theta(s)$ can be computed in the fundamental representation, where we have
\be
\tr\big(h\exp(s\tau_i)\big)=\f{\sin\big(2\theta(s)\big)}{\sin\theta(s)}=2\cos\theta(s),
\ee
and therefore
\be
\theta(s)=\arccos\left(\f{1}{2}\tr\big(h\exp(s\tau_i)\big)\right).
\ee
Using the chain rule, we can now rewrite \eqref{proof 1} as
\be
\tr_{(d)}\left(\d h\d\tau_i\right)=\f{\partial}{\partial\theta}\left(\f{\sin(d\theta)}{\sin\theta}\right)\left.\f{\partial\theta(s)}{\partial s}\right|_{s=0},
\ee
where $\theta=\theta(0)$. Finally, we can use the fact that
\be
\f{\partial\theta(s)}{\partial s}=\f{\partial}{\partial s}\arccos\left(\f{1}{2}\tr\big(h\exp(s\tau_i)\big)\right)=-\f{1}{2\sin\theta(s)}\f{\partial}{\partial s}\tr\big(h\exp(s\tau_i)\big)
\ee
to obtain formula \eqref{trace formula}.
\end{proof}

\section{Polynomial form of the curvature operator}
\label{appendixB}

\noindent In this appendix, we give a closed formula for the spin-dependent term appearing in the regularization of the curvature with arbitrary spin. First, in order to treat the two regularization schemes at once, let us denote the class angle by
\be\label{theta and m}
\theta=\arccos(\mathbf{x}),
\ee
where
\be
\mathbf{x}=
\left\{\begin{array}{l}
\displaystyle-1\leq\cos(\mb c)+\f{1}{2}\sin^2(\mb c)\leq1\phantom{0\leq\cos^2(\mb c)\leq1}\hspace{-2cm}\text{for the scheme HR},\\[8pt]
\displaystyle0\leq\cos^2(\mb c)\leq1\vphantom{\f{1}{2}}\phantom{-1\leq\cos(\mu c)+\f{1}{2}\sin^2(\mb c)\leq1}\hspace{-2cm}\text{for the scheme CR}.
\end{array}\right.
\ee
The function of interest is
\be
\mathcal{F}_d(\theta)\coloneqq-\f{3}{d(d^2-1)}\f{1}{\sin\theta}\f{\partial}{\partial\theta}\left(\f{\sin(d\theta)}{\sin\theta}\right)=\f{3}{d(d^2-1)}\f{1}{\sin^2\theta}\big(\cot\theta\sin(d\theta)-d\cos(d\theta)\big).
\ee
Now, because of \eqref{theta and m}, we have that
\be
\sin^2\theta=1-\mathbf{x}^2,\q\q\cos(d\theta)=T_d(\mathbf{x}),\q\q\f{\sin(d\theta)}{\sin\theta}=U_{d-1}(\mathbf{x}),
\ee
where $T_d(\mathbf{x})$ and $U_d(\mathbf{x})$ are the Chebyshev polynomials (in the variable $\mathbf{x}$) of first and second kind respectively. With this, we get that
\be\label{F intermediate}
\widetilde{\mathcal{F}}_d(\mathbf{x})\coloneqq\mathcal{F}_d\big(\theta(\mathbf{x})\big)
&=\f{3}{d(d^2-1)}\f{1}{\mathbf{x}^2-1}\big(dT_d(\mathbf{x})-\mathbf{x}U_{d-1}(\mathbf{x})\big)\nn\\
&=\f{3}{d(d^2-1)}\f{1}{\mathbf{x}^2-1}\big((d+1)T_d(\mathbf{x})-U_d(\mathbf{x})\big),
\ee
where we have used the recursion relation
\be
\mathbf{x}U_{d-1}(\mathbf{x})=U_d(\mathbf{x})-T_d(\mathbf{x}).
\ee
Now, the last step is to show that the division by $\mathbf{x}^2-1$ on the right-hand side of \eqref{F intermediate} can actually be carried out, implying that $\widetilde{\mathcal{F}}_d(\mathbf{x})$ itself is polynomial in $\mathbf{x}$. For this, we make use of one of the explicit expressions for the Chebyshev polynomials. This is given by
\be
T_d(\mathbf{x})=\sum_{i=0}^{\lfloor d/2\rfloor}{{d}\choose{2i}}(\mathbf{x}^2-1)^i\mathbf{x}^{d-2i},\q\q U_d(\mathbf{x})=\sum_{i=0}^{\lfloor d/2\rfloor}{{d+1}\choose{2i+1}}(\mathbf{x}^2-1)^i\mathbf{x}^{d-2i},
\ee
where $\lfloor\cdot\rfloor$ denotes the floor function. After a simple manipulation of the binomial coefficients, we get that
\be
(d+1)T_d(\mathbf{x})-U_d(\mathbf{x})=\sum_{i=0}^{\lfloor d/2\rfloor}2i{{d+1}\choose{2i+1}}(\mathbf{x}^2-1)^i\mathbf{x}^{d-2i},
\ee
which implies that
\be\label{x polynomial}
\widetilde{\mathcal{F}}_d(\mathbf{x})=\f{3}{d(d^2-1)}\sum_{i=1}^{\lfloor d/2\rfloor}2i{{d+1}\choose{2i+1}}(\mathbf{x}^2-1)^{i-1}\mathbf{x}^{d-2i},
\ee
and therefore also
\be\label{theta polynomial}
\mathcal{F}_d(\theta)=\f{3}{d(d^2-1)}\sum_{i=1}^{\lfloor d/2\rfloor}2i{{d+1}\choose{2i+1}}(-1)^{i-1}(\sin^2\theta)^{i-1}(\cos\theta)^{d-2i}.
\ee
The expression \eqref{x polynomial} is therefore always a polynomial in $\cos(\mb c)$ and/or $\sin(\mb c)$ (or equivalently in $\cos(\lambda b)$ and/or $\sin(\lambda b)$), depending on the regularization scheme.

Now, some standard properties of the Chebyshev polynomials can be used in order to obtain a bound on the function $\mathcal{F}_d(\theta)$. For $-1\leq\mathbf{x}\leq 1$ or $0\leq\theta\leq\pi$, we have
\be
\big|T_d(\mathbf{x})\big|=\big|T_d(\cos\theta)\big|\leq 1,\q\q\big|U_d(\mathbf{x})\big|=\big|U_d(\cos\theta)\big|\leq d+1.
\ee
Therefore, from \eqref{F intermediate} we can write the (crude but nonetheless useful) bound
\be\label{bound on F}
\big|\mathcal{F}_d(\theta)\big|\leq\f{6}{d(d-1)}\f{1}{\sin^2\theta},
\ee
where the right-hand side is a decreasing function of $\mathbb{N}\ni d\geq2$.

\end{document}